\documentclass[11pt,a4paper]{article}
\usepackage{times}

\usepackage{pgfplots}

\usepackage{url}
\usepackage{wrapfig}
\usepackage[hidelinks,pagebackref]{hyperref}
\usepackage[utf8]{inputenc}
\usepackage[small]{caption}
\usepackage{amsfonts}
\usepackage{graphicx}
\usepackage{mathdots}
\usepackage{color}
\usepackage{amsmath}
\usepackage{amssymb}
\usepackage[numbers]{natbib}
\usepackage{nicefrac}
\usepackage{amsthm}
\usepackage{booktabs}
\usepackage{paralist}
\usepackage{xcolor}
\usepackage{algorithm}
\usepackage{algorithmic}
\usepackage{cleveref}
\usepackage{thm-restate}
\newcommand{\appsymb}{$\bigstar$}
\usepackage{diagbox}
\usepackage{authblk}
\usepackage{fullpage}
\usepackage[inline]{enumitem}

\usepackage{subcaption}
\usepackage{todonotes}
\usepackage{kbordermatrix}
\usepackage{pgfplots}
\usepackage{tikz}

\newcommand{\pfnote}[1]{}
\newcommand{\nbnote}[1]{}

\newtheorem{definition}{Definition}
\newtheorem{example}{Example}

\newcommand{\calF}{\mathcal{F}}
\newcommand{\calP}{\mathcal{P}}

\newcommand{\pref}{\succ}
\newcommand{\pos}{{\mathrm{pos}}}
\newcommand{\dist}{{\mathrm{dist}}}

\newcommand{\reals}{{\mathbb{R}}}

\newcommand{\POS}{{{\mathrm{POS}}}}

\newcommand{\EMD}{{{\mathrm{EMD}}}}

\newcommand{\rID}{{{\mathrm{rID}}}}
\newcommand{\ID}{{{\mathrm{ID}}}}
\newcommand{\id}{{{\mathrm{id}}}}
\newcommand{\UN}{{{\mathrm{UN}}}}
\newcommand{\AN}{{{\mathrm{AN}}}}
\newcommand{\an}{{{\mathrm{an}}}}
\newcommand{\ST}{{{\mathrm{ST}}}}
\newcommand{\stt}{{{\mathrm{st}}}}

\newcommand{\relvot}{{{\mathrm{freq}}}}

\newcommand{\expswaps}{{{\mathrm{expswaps}}}}
\newcommand{\relswaps}{{{\mathrm{relswaps}}}}
\newcommand{\relphi}{{{\mathrm{rel}\hbox{-}\phi}}}

\definecolor{darkgreen}{rgb}{0,0.5,0}
\definecolor{darkpink}{rgb}{0.75,0.25,0.25}
\definecolor{RED}{rgb}{1,0,0}

\newcommand{\figurecut}[1]{}
\newcommand{\cutfornow}[1]{}
\title{Putting a Compass on the Map of Elections\thanks{An 
extended abstract of this article has been accepted for publication in the 
proceedings of IJCAI 2021 and COMSOC 2021.}}

\author[1]{Niclas Boehmer}
\author[2]{Robert Bredereck}
\author[3]{Piotr Faliszewski}
\author[1]{Rolf Niedermeier}
\author[4]{Stanisław Szufa}

\affil[1]{\small
  Technische Universit\"at Berlin, Algorithmics and Computational 
  Complexity\protect\\
  \{niclas.boehmer,rolf.niedermeier\}@tu-berlin.de}
\affil[2]{\small
  Humboldt-Universit\"at zu Berlin, 
  robert.bredereck@hu-berlin.de}
 \affil[3]{\small
  AGH University,
  faliszew@agh.edu.pl}
 \affil[4]{\small
  Jagiellonian University,
  stanislaw.szufa@uj.edu.pl}
\date{\today}

\begin{document}

\maketitle

\begin{abstract}
  Recently, Szufa et al.~[AAMAS~2020] presented a 
  ``map of elections'' that visualizes a 
  set of 800~elections generated from various statistical cultures.
  While similar elections are grouped together on this map,
  there is no obvious interpretation of the elections' positions.
  We provide such an interpretation
  by introducing four canonical ``extreme'' elections, acting
    as a compass 
  on the map. We use them 
  to analyze both a dataset provided by Szufa et al.\ and a number of
  real-life elections. In effect, we find a new variant of the Mallows
  model and show that it captures real-life scenarios
  particularly 
  well.
\end{abstract}

\section{Introduction}

\citet{szu-fal-sko-sli-tal:c:map} recently proposed a
technique for visualizing sets of ordinal elections---i.e., elections
where each voter ranks the candidates from the most to the least
appealing one---based on given distances between them.
They have applied 
this technique to 800 elections coming from various
\begin{wrapfigure}{r}{8.5cm}
  \centering
  \includegraphics[width=7.3cm]{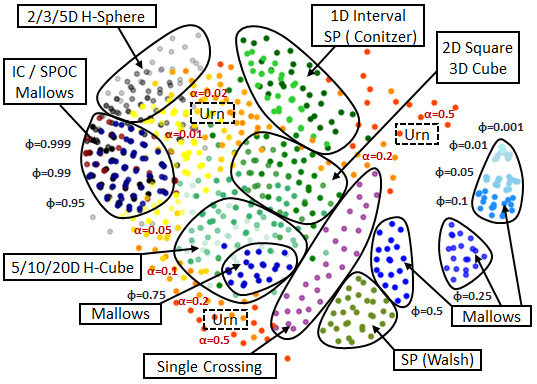}
  \caption{\label{fig:10x100-original}A map for the 10x100 dataset of \citet{szu-fal-sko-sli-tal:c:map}.} \vspace*{-0.2cm}
\end{wrapfigure}
statistical cultures, ranging from the classic urn and Mallows models
to various types of restricted domains, and they
obtained a \emph{map of elections}, where elections
with similar properties are grouped together
(see Figure~\ref{fig:10x100-original}; each dot represents a single
election and, generally, the closer two elections are in the picture,
the smaller is their distance in terms of the metric of Szufa et al.).
Indeed, we see that elections from the same statistical culture,
represented with the same color, are nicely grouped together;
\citet{szu-fal-sko-sli-tal:c:map} 
have also shown other evidence that nearby elections are closely
related.\footnote{The map in the figure regards elections with
  10~candidates and 100~voters, whereas
  \citet{szu-fal-sko-sli-tal:c:map} focused on the case of
  100~candidates and 100~voters. Nonetheless, they also provided such
  smaller datasets and we focus on them because we want to compare
  them to real-life elections, which typically have few candidates.}
Yet, the map 
has two major drawbacks. First, while similar elections are plotted
next to each other, there is no clear meaning to absolute
positions on the map. Second, the map regards
statistical cultures only and it is not obvious where real-life
elections---such as those stored in
PrefLib~\citep{mat-wal:c:preflib}---would lie on the map. Our goal
is to address both these issues.

\begin{figure}
    \begin{subfigure}[b]{0.49\textwidth}
    \centering            
    \includegraphics[width=6.7cm]{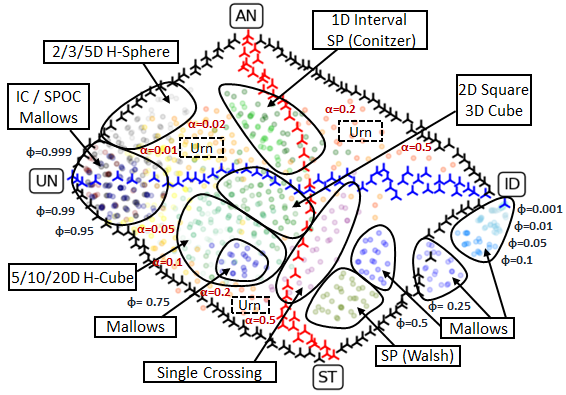}
    \caption{\label{fig:10x100SC}\footnotesize The compass and the dataset of \citet{szu-fal-sko-sli-tal:c:map}.}
  \end{subfigure}
  \hfill
   \begin{subfigure}[b]{0.49\textwidth}
    \centering            
    \includegraphics[width=6.7cm]{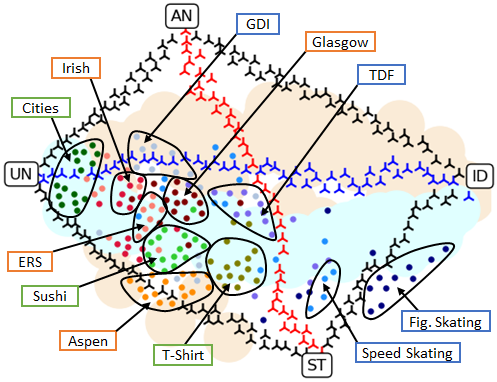}
    \caption{\label{fig:10x100RW}\footnotesize 
    The compass and various real-life elections. }
  \end{subfigure}
  \caption{\label{fig:10x100+compass}Maps of elections for the 10x100
    dataset of \citet{szu-fal-sko-sli-tal:c:map} (on the left) and
    some real-life elections (on the right).  The maps include the
    compass matrices and their connecting paths (shown in black, red,
    and blue).  For clarity, the dots corresponding to the elections
    of \citet{szu-fal-sko-sli-tal:c:map} (in the left figure) are
    shown in lighter colors than in Figure~\ref{fig:10x100-original}.
    On the right, the pale blue area is where Mallows elections end up
    (for various $\phi$ parameters) and the pale orange area is where urn
    elections end up (for various $\alpha$ parameters).}
\end{figure}

We start by looking more closely at the distance metric for elections  that \citet{szu-fal-sko-sli-tal:c:map} used.  The idea is
that given an election with $m$~candidates, one computes an
$m \times m$ \emph{frequency matrix} which specifies what fraction of
the voters ranks each candidate in each position (such matrices are
bistochastic, i.e., their entries are nonnegative and each column and
each row sums up to one).
Measuring the distance between two elections boils
down to computing their frequency matrices and summing up the earth
mover's distances between their columns, where columns are reordered  to minimize the
distance (see Section~\ref{sec:prelim}); \citet{szu-fal-sko-sli-tal:c:map} call this distance \emph{positionwise}.
Using frequency matrices makes it possible to compare elections with
different numbers of voters (effectively, by reweighting them), and
reordering the columns ensures that candidate names are irrelevant (as
suggested for such settings by
\citet{fal-sko-sli-szu-tal:c:isomorphism}).

Our first result is an algorithm that, given a bistochastic matrix and
a number~$n$, finds \emph{some} election with $n$~voters whose
frequency matrix is (nearly) identical to the one from the input (achieving
perfect accuracy 
is not always possible,
but our algorithm achieves the best result one may, in general, hope for).  As
a consequence, instead of considering elections, we may directly look
at the space of bistochastic matrices, which simplifies many
discussions.
Thus, we often speak of matrices and elections interchangeably.

Next we form what we call a \emph{compass}. The idea is to pick
matrices that, on the one hand, are far away from each other, and, on
the other hand, have natural interpretations. Specifically, we
consider the following four ``extreme'' matrices, corresponding to four types of
(dis)agreement among the voters:
\begin{enumerate}
\item The identity matrix, $\ID$, modelling 
  perfect agreement (all voters agree on a single
  preference order).
\item The uniformity matrix, $\UN$, modelling
  lack of agreement (each candidate takes each position equally often).
\item The stratification matrix, $\ST$, modelling partial agreement (voters agree
  that half of the candidates are better than the other half, but
  lack agreement on anything else).
\item The antagonism matrix, $\AN$, 
  modelling
  conflict (half of the voters have
  opposite preference orders to the other half). 
\end{enumerate}
For each two of these ``compass'' matrices, we also consider a
spectrum of their convex combinations (``paths'' between the
matrices). In visualizations, these paths appear as a
parallelogram-like shape with corresponding ``diagonals''; see, e.g.,
Figure~\ref{fig:10x100SC}, where we apply the compass method to the
dataset from Figure~\ref{fig:10x100-original} (the black, blue, and
red points are certain convex combinations of the corresponding
endpoints, which are the compass matrices).
Throughout the rest of the paper we explain where these figures come from.

The compass 
allows us to make several observations. For example, in
Figure~\ref{fig:10x100SC} we see that 1D Interval elections are closer
to the antagonism matrix, whereas higher-dimensional hypercube
elections are closer to the stratification one. This is intriguing as,
on a formal level, these two kinds of elections are very similar.
Figure~\ref{fig:10x100RW}, which shows a map of real-life elections
(some 
from PrefLib and some new ones) is even more striking. Most of the real-life
elections (including all political ones) end up in one ``quadrant'' of
the parallelogram, and essentially all elections end up in the
vicinity of some Mallows elections (in Figure~\ref{fig:10x100RW}, the pale blue area is where
Mallows elections end up, depending on the parameter of the model; the
pale orange area is where urn elections end up).  So,
if one were to run experiments with a single statistical culture, 
the Mallows model might be a wise choice. 

Yet, 
we find that natural ways
of sampling Mallows elections (e.g., by
choosing the Mallows 
parameter uniformly at random, or using a fixed parameter 
for different numbers of candidates), which are used in many research
papers, are biased.
We propose a normalization and argue, both
theoretically and by considering the compass, that it produces more
balanced results.  In other words, we recommend using the Mallows
model, but in conjunction with our normalization.

We provide a detailed analysis and discussion of the above-mentioned results in
the main part of the paper and in the appendix (results marked by (\appsymb) are proven in the appendix.).

\section{Preliminaries}\label{sec:prelim}
Given an integer~$t$, we write~$[t]$ to denote the set
$\{1, \ldots, t\}$. By~$\reals_+$ we mean the set of nonnegative real
numbers.
Given a vector $x = (x_1, \ldots, x_m)$, we interpret it as an
$m \times 1$ matrix, i.e., we use column vectors.
For a matrix $X$, we write $x_{i,j}$ to refer to the entry in its $i$-th
row and $j$-th column.\smallskip

\noindent\textbf{Elections.}
An election $E$ is a pair $(C,V$), where $C = \{c_1, \ldots, c_m\}$ is
a set of candidates and $V = (v_1, \ldots, v_n)$ is a collection of
voters. Each voter $v \in V$ has a preference order $\pref_v$, which
ranks the candidates from the most to the least desirable one
according to~$v$. If $v$ prefers candidate~$a$ to candidate~$b$, then
we write $v \colon a \pref b$, and we extend this notation to more
candidates in a natural way.
For a voter $v$ and a candidate $c$, we write $\pos_v(c)$ to denote
the position on which $v$ ranks $c$ (the top-ranked candidate has
position $1$, the next one has position $2$, and so on).
We refer to both the voters and their preference orders as the
votes. The intended meaning will always be clear from the context.
\smallskip

\noindent\textbf{Position and Frequency Matrices.}
Let $E = (C,V)$ be an election, where $C = \{c_1, \ldots, c_m\}$ and
$V = (v_1, \ldots, v_n)$.  For a candidate $c \in C$ and position
$i \in [m]$, we write~$\#\pos_E(c,i)$ to denote the number of
voters in election $E$ that rank $c$ on position~$i$. By~$\#\pos_E(c)$
we mean the vector:
\[
  ( \#\pos_E(c,1), \#\pos_E(c,2), \ldots, \#\pos_E(c,m) ).
\]
The \emph{position matrix} for election $E$, 
denoted $\#\pos(E)$, is the $m \times m$ matrix that has  vectors~$\#\pos_E(c_1), \ldots, \#\pos_E(c_m)$ as its columns.
We also consider
vote frequencies rather than absolute counts.  To this end, for a
candidate $c$ and a position $i \in [m]$, let $\#\relvot_E(c,i)$
be~$\frac{\#\pos_E(c,i)}{n}$, let vector $\#\relvot_E(c)$
be~$ ( \#\relvot_E(c,1), \ldots, \#\relvot_E(c,m) ), $
and let the \emph{frequency matrix} for election $E$, denoted~$\#\relvot(E)$, consist of columns~$\#\relvot_E(c_1), \ldots, \#\relvot_E(c_m)$.

Note that in each position matrix, each row and each column
sums up to the number of voters in the election. Similarly, in each
frequency matrix, the rows and columns sum up to one (such matrices
are called bistochastic).
For a positive integer $m$, we write $\calF(m)$ [$\calP(m)$] to denote the set of
all $m \times m$ frequency [position] matrices.

\begin{example}
  Let $E = (C,V)$ be an election, where 
  $C = \{a,b,c\}$,  $V = (v_1, \ldots, v_6)$, and the preference orders are
  $v_1 \colon a \pref b \pref c$,
  $v_2 \colon a \pref b \pref c$,
  $v_3 \colon a \pref b \pref c$,
  $v_4 \colon b \pref a \pref c$,
  $v_5 \colon c \pref a \pref b$,
  $v_6 \colon c \pref a \pref b$. 
  The position and frequency matrices for this election are:
  \begin{align*}
     \kbordermatrix{ & a & b & c  \\
    1 &                3 & 1 & 2  \\
    2 &                3 & 3 & 0  \\
    3 &                0 & 2 & 4   
    }
                               \text{\quad and\quad}
     \kbordermatrix{ & a & b & c  \\
    1 &                \nicefrac{1}{2}  & \nicefrac{1}{6}  & \nicefrac{1}{3} \\
    2 &                \nicefrac{1}{2}  & \nicefrac{1}{2}  & 0  \\
    3 &                              0  & \nicefrac{1}{3}  & \nicefrac{2}{3}  
    }
  \end{align*}  
\end{example}
\smallskip

\noindent\textbf{Earth Mover's Distance (EMD).}
Let $x = (x_1, \ldots, x_t)$ and $y = (y_1, \ldots, y_t)$ be two
vectors from $\reals_+^t$, whose entries sum up to $1$.
The \emph{earth mover's distance} between $x$ and~$y$, denoted $\EMD(x,y)$,
is defined as the lowest total cost of operations that transform
vector~$x$ into vector $y$, where each operation is of the form
``\emph{subtract $\delta$ from
position~$i$ and add $\delta$ to position~$j$}'' and costs
$\delta \cdot |i-j|$. Such an operation is legal if the current value
at position~$i$ is at least $\delta$.
$\EMD(x,y)$ can be computed in polynomial time using a
greedy algorithm.
\smallskip

\noindent\textbf{Positionwise Distance.}
Let $E = (C,V)$ and $F = (D,U)$ be two elections
with $m$ candidates each
(we do not require that $|V| = |U|$). 
The positionwise distance between~$E$ and~$F$,
denoted $\POS(E,F)$, is defined
in terms of 
frequency matrices
 $\#\relvot(E) = (e_1, \ldots, e_m)$ 
and~$\#\relvot(F) = (f_1, \ldots, f_m)$ as follows \citep{szu-fal-sko-sli-tal:c:map}:
\[
  \textstyle
  \POS(E,F) := \min_{\sigma \in S_m} \left( \sum_{i=1}^m \EMD(e_i, f_{\sigma(i)}) \right),
\]
where $S_m$ is the permutation group on $m$~elements.
In other words, the positionwise distance is the sum of the
 earth mover's
distances between the frequency vectors of the candidates from the two
elections, with candidates/columns matched optimally according to $\sigma$.
The positionwise distance is invariant to renaming the candidates and
reordering the voters. 
\smallskip

\noindent\textbf{Statistical Cultures.}
We define the following
three statistical cultures, i.e., models for generating random
elections:

\begin{enumerate}
\item
  Under the Impartial Culture (IC) model,
  we sample all votes uniformly at random.
\item
  The Pólya-Eggenberger urn model~\citep{berg1985paradox}
  uses 
  a nonnegative parameter $\alpha$, which gives the level of
  correlation between the votes (this parameterization is due to
  \citet{mcc-sli:j:similarity-rules}). To generate an election with
  $m$ candidates, we take an urn containing
  one copy of each possible preference order and generate
  the votes iteratively: 
  In each step we draw an order from the urn (this is the newly
  generated vote) and return it to the urn together
  with $\alpha m!$ copies. For $\alpha = 0$, we obtain the IC model.
\item The Mallows model~\citep{mal:j:mallows} uses parameter
  $\phi \in [0,1]$ and a central preference order $v$.  Each vote is
  generated randomly and independently. The probability of generating
  a vote~$u$ is proportional to $\phi^{\kappa(u, v)}$, where
  $\kappa(u, v)$ is the swap distance between $u$ and $v$ (i.e., the
  minimum number of swaps of adjacent candidates that transform
  $u$ into~$v$).

\end{enumerate}

Sometimes, we refer to other statistical cultures used by
\citet{szu-fal-sko-sli-tal:c:map}. We do not define them formally
here, but we attempt to make our discussions intuitively
clear.\smallskip

\noindent\textbf{Maps of Elections.}
\citet{szu-fal-sko-sli-tal:c:map} drew a \emph{map of elections} by
computing the positionwise distances between 800 elections drawn
from various statistical cultures and visualizing them using the
force-directed algorithm of~\citet{fruchterman1991graph}. 
They focused on elections with 100 candidates and 100 voters, but also
generated smaller datasets, available on their website.  We consider
their dataset with 10 candidates and 100 voters (see
Figure~\ref{fig:10x100-original} for its map).  We use the same
algorithm as they do for our visualizations, except that for each two
elections we set their attraction coefficient to be the square of
their positionwise distance (and not the distance itself, as they do;
our approach groups similar elections more tightly and gives more
visually appealing results for elections with few candidates).

We stress that the maps that both we and Szufa et
al.~\cite{szu-fal-sko-sli-tal:c:map} provide are helpful tools to
illustrate the distances between particular (families of) elections,
but are certainly not perfect. For example, since the visualization
algorithm is randomized, we can get slightly different maps for each
run of the algorithm. The visualizations also depend on the exact
composition of the set of depicted elections (for example, a map where
50\% of the elections came from the IC model would make it seem that
these elections cover a much larger proportion of the map than if IC
elections constituted only 10\% of the elections).  Thus, whenever we
say that some two elections are close to each other, we mean that
their positionwise distance is small. While this is typically
reflected by these two elections being close on the map, on its own,
closeness on the map does not suffice for such a claim.

\section{Recovering Elections from Matrices}\label{se:recov}

Throughout this paper we often deal with frequency matrices of
elections. While computing a frequency matrix of an election is
straightforward, 
the reverse direction is less clear.

We first observe that each $m\times m$ position matrix has a corresponding
$m$-candidate election with at most $m^2-2m+2$ distinct preference
orders. This was shown by \citet[Theorem 7]{leep1999marriage} (they
speak of ``semi-magic squares'' and not ``position matrices'' and
show a decomposition of a matrix into permutation matrices, which
correspond to votes in our setting).
Their proof
lacks some algorithmic details which we provide in the appendix.

\begin{restatable}[\appsymb]{proposition}{polyalgo}\label{pro:poly-algo}
  Given a position matrix $X\in \calP(m)$, one can compute
  in $O(m^{4.5})$ time  an election~$E$ that contains at most $m^2-2m+2$
  different votes such that $\#\pos(E) =
  X$. 
\end{restatable}

Next, we consider the issue of recovering elections based on frequency
matrices. Given an $m \times m$ bistochastic matrix $X$ and a number
$n$ of voters, we would like to find an election~$E$ with position
matrix $nX$. This may be impossible as $nX$ may have fractional
entries, but we can get very close to this~goal.
The next proposition shows how to achieve it, and justifies speaking
of elections and frequency matrices interchangeably.

\begin{restatable}{proposition}{freqtopos}\label{pro:frequency}
  Given an $m \times m$ bistochastic matrix $X$ and an integer~$n$,
  one can compute in polynomial time an election~$E$ with $n$ voters
  whose position matrix~$P$ satisfies $|nx_{i,j} - p_{i,j}| \leq 1$
  for each $i, j \in [m]$ and, under this condition, minimizes the
  value $\sum_{1 \leq i,j \leq m} |nx_{i,j} - p_{i,j}|$.
\end{restatable}
\begin{proof}
  We use randomized dependent rounding in the following algorithm; see
  Appendix~A for a deterministic algorithm, which also performs the
  minimization step.

  We start by computing matrix $Y$
  where each entry $y_{i,j}$ is equal to
  $nx_{i,j} - \lfloor nx_{i,j}\rfloor$.
  All entries of $Y$ are between~$0$
  and~$1$, and each row and each column
  of~$Y$ sums up to a (possibly different) integer.  
  We construct an edge-weighted bipartite
  graph~$G$ with vertex sets~$A = \{a_1, \ldots, a_m\}$ and~$B =
  \{b_1, \ldots, b_m\}$. For each two vertices~$a_i$
  and~$b_j$, we have a connecting edge of
  weight~$y_{i,j}$. For each vertex $c \in A \cup
  B$, we let its fractional degree~$\delta_G(c)$ be the sum of the weights of the edges touching it.
  Then we invoke the dependent rounding procedure of
  \citet{gan-khu-par-sri:j:dependent-rounding} on this graph, and in
  polynomial time we obtain an unweighted bipartite graph
  $G'$ with the same two vertex sets, such that the (standard) degree
  of each vertex $c \in A \cup B$ in $G'$ is equal to $\delta_G(c)$
  (note that dependent rounding is computed via a randomized
  algorithm, but this condition on the degrees is always satisfied,
  independently of the random bits selected).
  Using $G'$, we form an $m \times m$ matrix $D$ such that for each $i,j \in [m]$,~$d_{i,j}$ is $1$ if $G'$ contains an edge between $i$ and $j$, and
  $d_{i,j} = 0$ otherwise.
  Finally, we compute matrix $P = \lfloor nX \rfloor + D$.
   
  The entries of~$P$ differ from those of~$nX$ by at most one, and the
  rows and columns of $P$ sum up to $n$ (to see it, consider the
  degrees of the vertices in~$G'$).
  So, we obtain the desired election by invoking
  Proposition~\ref{pro:poly-algo} on matrix $P$.
\end{proof}

\section{Setting up the Compass}\label{se:setcomp}
Our ``compass'' consists of two main components: Four matrices that
occupy very different areas of the election space 
and represent
different types of (dis)agreement among the voters, and six paths
consisting of their convex combinations.
We describe these components below.

\subsection{The Four Matrices}\label{sec:four}

The first two matrices are the \emph{identity} matrix, $\ID_m$, with
ones on the diagonal and zeros elsewhere, and the \emph{uniformity}
matrix, $\UN_m$, with each entry equal to $\nicefrac{1}{m}$.
The identity matrix corresponds to elections where each voter has the
same preference order, i.e., there is a common ordering of the
candidates from the most to the least desirable one. In contrast,
the uniformity matrix captures elections where each candidate is
ranked on each position equally often, i.e., where, in aggregate, all
the candidates are viewed as equally good. Uniformity elections are
quite similar to the IC ones
and, in the limit, indistinguishable from them.
Yet, for a fixed number of voters, typically IC elections are at some
(small) positionwise distance from uniformity.
\medskip

The next matrix, \emph{stratification}, is defined as
follows (we assume that $m$ is even):
\[
  \ST_m = \begin{bmatrix}
    \UN_{\nicefrac{m}{2}} & 0 \\
    0 & \UN_{\nicefrac{m}{2}}
  \end{bmatrix}.
\]
Stratification matrices correspond to elections where the voters agree
that half of the candidates are 
more desirable than the other half, but, in aggregate, are unable to
distinguish between the qualities of the candidates in each group.
\medskip

For the next matrix, we need one more piece of notation.
Let $\rID_m$ be the matrix obtained by reversing the order of the columns of the identity matrix $\ID_m$. We define the \emph{antagonism} matrix,
$\AN_m$, to be
$
\textstyle
\nicefrac{1}{2} \ID_m+\nicefrac{1}{2} \rID_m.
$
Such matrices are generated, e.g., by
elections where half of the voters rank the
candidates in one order, and half of the voters rank them in the
opposite one, so there is a clear conflict. 
In some sense, stratification and antagonism are based on
similar premises. Under stratification, the group of candidates is
partitioned into halves with different properties, whereas in
antagonism (for the case where half of the voters rank the candidates in the same order)  the voters are partitioned. However, the nature of the
partitioning is, naturally, quite different.\medskip

We chose the above matrices
because they capture natural, intuitive
phenomena and 
seem to occupy very different areas of the
space of elections.  To see that the latter holds, let us calculate
their positionwise distances (for further arguments see also Appendix~B).

\begin{restatable}[\appsymb]{proposition}{dist}\label{pr:calc}
    If $m$ is divisible by $4$, then it holds that:
  \begin{enumerate}
  \item $\POS(\ID_m,\UN_m) = \frac{1}{3}(m^2-1)$,
  \item $\POS(\ID_m,\AN_m) = \POS(\UN_m,\ST_m) = \frac{m^2}{4}$,
  \item
    $\POS(\ID_m,\ST_m) = \POS(\UN_m,\AN_m)  = 
    \frac{2}{3}(\frac{m^2}{4}-1)$,
  \item $\POS(\AN_m,\ST_m) = \frac{13}{48} m^2 - \frac{1}{3}$.
  \end{enumerate}
\end{restatable}

\noindent
To normalize these distances, we divide them by
$D(m) = \POS(\ID_m, \UN_m)$, which we suspect to be the largest
positionwise distance between two matrices from $\calF(m)$ (as argued
below).  For each two matrices $X$ and $Y$ among our four, we let
$d(X,Y) := \lim_{m \rightarrow
  \infty}\nicefrac{POS(X_m,Y_m)}{D(m)}$. A simple computation shows
the following (see also the drawing on the right; we sometimes omit
the subscript $m$ for simplicity):

\begin{minipage}[b]{0.45\textwidth}
  \centering
  \begin{align*}
  &d(\ID,\UN) = 1,\\
  &d(\ID,\AN) = d(\UN, \ST) = \nicefrac{3}{4},\\
  &d(\AN,\ST) = \nicefrac{13}{16},\\
  &d(\ID,\ST) = d(\UN,\AN) = \nicefrac{1}{2}.\\
\end{align*}
\end{minipage}
\begin{minipage}[b]{0.5\textwidth}
  \centering
        \newcommand{\drawun}[2]{
    \draw (#1+0.5, #2+1) node[anchor=south] {UN};
    \fill[black!25!white] (#1+0,#2+0)  rectangle (#1+1,#2+1);
    \draw (#1+0,#2+0) rectangle (#1+1,#2+1);
  }

  \newcommand{\drawan}[2]{
    \draw (#1+1.75, #2+0.5) node[anchor=south] {AN};
    \fill[black!25!white] (#1+0,#2+0)  -- (#1+0.2, #2+0) -- (#1+1, #2+1-0.2) -- (#1+1, #2+1) -- (#1+1-0.2, #2+1) -- (#1, #2+0.2) -- cycle;
    \fill[black!25!white] (#1+0,#2+1)  -- (#1+0.2, #2+1) -- (#1+1, #2+0.2) -- (#1+1, #2) -- (#1+1-0.2, #2) -- (#1, #2+1-0.2) -- cycle;
    \draw (#1+0,#2+0) rectangle (#1+1,#2+1);
  }

  \newcommand{\drawid}[2]{
    \draw (#1+0.5, #2+1) node[anchor=south] {ID};
    \fill[black!25!white] (#1+0,#2+1)  -- (#1+0.2, #2+1) -- (#1+1, #2+0.2) -- (#1+1, #2) -- (#1+1-0.2, #2) -- (#1, #2+1-0.2) -- cycle;
    \draw (#1+0,#2+0) rectangle (#1+1,#2+1);
  }

  \newcommand{\drawst}[2]{
    \draw (#1-0.75, #2-0.5) node[anchor=south] {ST};
    \fill[black!25!white] (#1+0,#2+1)  rectangle (#1+0.5, #2+0.5);
    \fill[black!25!white] (#1+0.5,#2+0.5)  rectangle (#1+1, #2+0);
    \draw (#1+0,#2+0) rectangle (#1+1,#2+1);
  }

    \begin{tikzpicture}[xscale=0.5, yscale=0.5]
    \clip (-0.1, -3) rectangle (9, 3.5);
    \drawun{0}{0}
    \drawid{8}{0}
    \drawan{3}{2}
    \drawst{5}{-2}
    \draw (1,0.5) -- (8,0.5);
    \draw (3,0.5) node[anchor=south] {$1$};
    \draw (1,1) -- (3,2.5);
    \draw (2,1.75) node[anchor=south] {$\frac{1}{2}$};
    \draw (4,2.5) -- (8,1);
    \draw (6,1.75) node[anchor=south] {$\frac{3}{4}$};
    \draw (4,2) -- (5,-1);
    \draw (4.7,0.75) node[anchor=south] {$\frac{13}{16}$};
    \draw (1,0) -- (5,-1.5);
    \draw (2,-0.5) node[anchor=north] {$\frac{3}{4}$};
    \draw (6,-1.5) -- (8,0);
    \draw (7.2,-0.5) node[anchor=north] {$\frac{1}{2}$};
  \end{tikzpicture}
\end{minipage}

For small~$m$, using ILPs, 
we verified that each compass matrix is almost as far away as possible
from the others.  Further, we believe that $\ID$ and $\UN$ are the two
most distant frequency matrices, i.e., they form the diameter of our
space. While showing this formally seems to be challenging, for each
$m \in \{3, \ldots, 7\}$, using an ILP, we have verified that, indeed,~$\ID_m$ and $\UN_m$ are the two most distant matrices under the
positionwise distance.

\subsection{Paths between Election Matrices}

Next, we consider convex combinations of frequency matrices.
Given two such matrices,~$X$ and $Y$, and $\alpha \in [0,1]$, one
might expect that matrix $Z = \alpha X + (1-\alpha)Y$ would lie at
distance~$(1-\alpha) \POS(X,Y)$ from $X$ and at distance
$\alpha \POS(X,Y)$ from~$Y$, so that we would~have:
\begin{align*}
     \POS(X,Y) = \POS(X, Z) 
                      + \POS(Z, Y).
\end{align*}
However, without further assumptions this is not necessarily the
case. Indeed, if we take~$X = \ID_m$ and $Y = \textrm{rID}_m$, then
$0.5X+0.5Y = \AN_m$ and $\POS(X,Y) = 0$, but
$\POS(X,0.5X+0.5Y) = \POS(\ID,\AN) > 0$.
Yet, if we arrange the two 
matrices $X$ and $Y$
so that their positionwise distance is achieved by the identity
permutation of their column vectors, then their convex combination lies
exactly between them. 
\begin{restatable}[\appsymb]{proposition}{paths}\label{pro:paths}
     Let $X = (x_1, \ldots, x_m)$ and $Y = (y_1, \ldots y_m)$ be two
  $m \times m$ frequency matrices such that
  $
     \POS(X,Y) = \textstyle \sum_{i=1}^m \EMD(x_i,y_i).
  $
  Then, for each $\alpha \in [0,1]$ it holds that~$\POS(X,Y) = \POS(X, \alpha X + (1-\alpha)Y) + \POS( \alpha X +
  (1-\alpha)Y, Y)$.
\end{restatable}
Using  Proposition~\ref{pro:paths}, for each two compass matrices,
we can generate a sequence of matrices that form a path between them.
For example, matrix $0.5\ID + 0.5\UN$ is exactly at the same distance
from $\ID$ and from $\UN$.  In Figure~\ref{fig:10x100SC} we show a map
of elections that (in addition to the dataset of
\citet{szu-fal-sko-sli-tal:c:map}) contains our four compass matrices
and for each two of them, i.e., for each two
$X, Y \in \{\ID, \UN, \AN, \ST\}$, a set of
$\lceil 50\cdot d(X,Y) \rceil$ matrices obtained as their convex
combinations with values of~$\alpha$ uniformly distributed in $[0,1]$.
Note that by the proof of Proposition~\ref{pr:calc}, it holds that the positionwise
distance between any two of our four matrices is achieved for the
identity permuation, as required by Proposition~\ref{pro:paths}.

\section{Applying the Compass}\label{se:appcomp}
In this section, we apply our compass to gain a better understanding
of the map of elections created by \citet{szu-fal-sko-sli-tal:c:map}
and to place some real-life elections on the map.  We also determine
where Mallows and urn elections land.

\subsection{A Map of Statistical Cultures with a Compass}

In Figure~\ref{fig:10x100SC}, we show a map of the 800 elections
provided by \citet{szu-fal-sko-sli-tal:c:map} in their 10x100 dataset,
together with the compass.
As expected, the uniformity matrix is close to the impartial culture
elections, but still at some distance from them.
Similarly, the identity matrix is very close to the Mallows
elections with close-to-zero values of~$\phi$. Indeed, such elections
consist of nearly identical votes.

The red path, linking $\AN$ and $\ST$, roughly partitions the
elections into those closer to~$\UN$ and those closer to $\ID$. The
latter group consists mostly of Mallows and urn elections (with low
$\phi$ or high $\alpha$, respectively), but single-crossing and some
single-peaked elections also make an appearance.

Analyzing the distances of elections to $\AN$ and $\ST$, it is
striking that 1D Interval elections lie closer to $\AN$, while other
hypercube elections lie closer to $\ST$, even though, formally, they
are similar.\footnote{Elections in a $t$-dimensional hypercube model
  are generated by drawing, for each candidate and each voter, a point
  from $[0,1]^t$ uniformly at random.  A voter then ranks the candidates with respect to
  the increasing distance of their points from his or her. For
  $t \in \{1,2,3\}$, $t$-dimensional hypercube elections are called 1D
  Interval, 2D Square, and 3D Cube, respectively. The others are called
  $t$D H-Cube.} Moreover, it is intriguing that single-peaked
elections generated according to the Walsh
model~\citep{wal:t:generate-sp} are closer to $\ST$, whereas those
from the Conitzer model~\citep{con:j:eliciting-singlepeaked} (which
are very similar to the $1$D Interval ones) are closer to $\AN$.  To
understand this phenomenon, let us look at single-peakedness and the
models of Conitzer and Walsh more closely.

\begin{definition}
  Consider a set $C = \{c_1, \ldots, c_m\}$ of candidates and a linear
  order $c_1 \lhd c_2 \lhd \cdots \lhd c_m$, referred to as the
  \emph{societal axis}. A vote $v$ is single-peaked with respect to
  $\lhd$ if for each $t \in [m]$ it holds that the $t$ top-ranked
  candidates in $v$ form an interval in $\lhd$. An election
  $E = (C,V)$ is single-peaked if there is a societal axis $\lhd$ such
  that each vote in $V$ is single-peaked with respect to $\lhd$.
\end{definition}

Intuitively, the societal axis orders the candidates with respect to
some common, one-dimensional issue, such as, e.g., their position on the
political left-to-right spectrum.
In both the Conitzer and the Walsh model, we start by choosing the
axis uniformly at random.
Then, in the Conitzer model,
we generate each vote as follows: We choose the top-ranked candidate
uniformly at random and we keep on extending the vote with candidates
to the left and to the right of the already-selected ones, deciding
which one to pick with a coin toss, until the vote is complete.
Thus, by choosing close-to-extreme candidates from different sides of
the axis as top-ranked, we generate close-to-opposite preference
orders with fairly high probability. As a consequence, the Conitzer
model generates elections that have common features with the
antagonism ones.
Under the Walsh model, we choose each single-peaked preference order
uniformly at random.  There are few such preference orders with
extreme candidates (with respect to the axis) ranked highly, but many
with the center candidates on top and the extreme candidates ranked around the bottom. This leads to 
stratification.

\begin{figure}
  \centering
  \begin{subfigure}[t]{4.15cm}
    \centering
    \includegraphics[width=3.9cm]{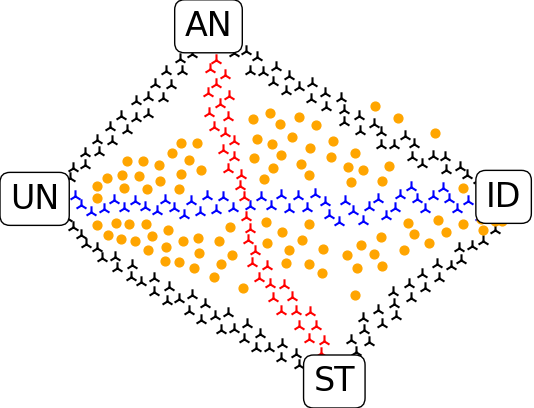}
    \caption{\label{fig:urns} Urn elections (orange); $\alpha$~follows the Gamma distribution.}
  \end{subfigure}
  \quad
  \begin{subfigure}[t]{4.15cm}
    \centering
    \includegraphics[width=3.9cm]{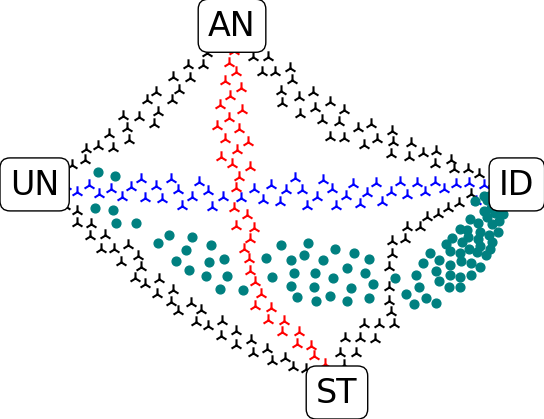}
    \caption{\label{fig:mallows-phi} Mallows elections (teal); $\phi$ follows the uniform distribution.}
  \end{subfigure}
  \quad
  \begin{subfigure}[t]{4.15cm}
    \centering
    \includegraphics[width=3.9cm]{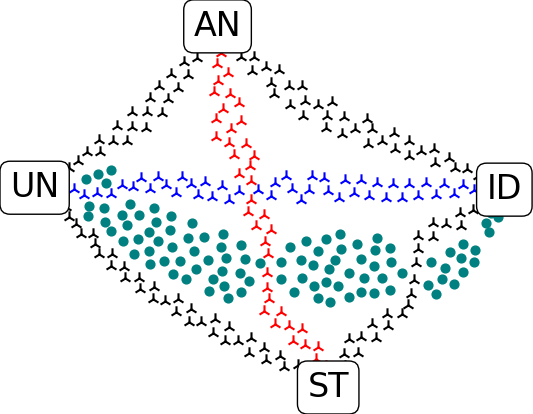}
    \caption{\label{fig:mallows-rel} Mallows elections (teal); $\relphi$ follows the uniform distribution.}
  \end{subfigure}
  \quad

  \caption{\label{fig:mallows-urns}
    Maps showing the four compass
    matrices, their connecting paths, and, respectively, urn
    elections and Mallows elections (for two distributions of their parameter).
    These visualizations are for 20 candidates and 100 voters.}
\end{figure}

\subsection{Urn and Mallows Elections}

Our next goal is to place ``paths'' of urn and Mallows elections on
the map. In both cases it requires some care. Recall that the urn
model has parameter $\alpha$, which takes values between~$0$ and $\infty$. To generate an urn election, we choose
$\alpha$ according to the Gamma distribution with shape parameter
$k=0.8$ and scale parameter $\theta=1$ (this ensures that about half of
the urn elections are closer to $\UN$ than to $\ID$, and the opposite
for the other half; see Figure~\ref{fig:urns}).

Regarding the Mallows model, 
we have a parameter $\phi$ that takes values between $0$ and~$1$,
where $0$ leads to generating $\ID$ elections and $1$ leads to
generating IC ones.  It is thus intuitive to choose $\phi$ uniformly
at random from the $[0,1]$ interval.
Yet, as seen in
\Cref{fig:mallows-phi}, doing so places elections quite unevenly on
the map.
Similarly, for different numbers of candidates the same value of
$\phi$ leads to choosing elections at different distances from $\ID$, which thereby also end up in (quite) different places on the map
(see the top-left part of Figure~\ref{fig:mallows-issues}).
This imbalance, in essence, follows from the fact that making, say, five swaps among 10 candidates has a stronger effect than making five swaps among 20 candidates.
\smallskip

\noindent\textbf{Normalizing Mallows.}
Let us consider a setting with $m$ candidates. For a $\phi \in [0,1]$,
let $\expswaps(m,\phi)$ be the expected swap distance between an
$m$-candidate vote generated using the Mallows model with parameter
$\phi$ and the center vote.  We define the relative expected number of
swaps~as:
\[
  \textstyle \relswaps(m,\phi) = \frac{\expswaps(m,\phi)}{\nicefrac{m(m-1)}{2}}
\]
(see the bottom-left part of Figure~\ref{fig:mallows-issues} for plots
of this value, depending on $\phi$ and $m$).  
In our approach, we choose a value $\relphi \in [0,1]$ as a
parameter, 
find a $\phi$ such that
$\relswaps(m,\phi) = \relphi$, and 
draw an election from the
Mallows model using this $\phi$
(see Appendix~C.1 for details).
Working on $\relphi$ instead of
$\phi$ not only allows for an intuitive and natural interpretation of
the parameter as the relative expected number of swaps in each vote (or the normalized distance from $\ID$), but also for obtaining comparable elections for
different numbers of candidates.
In Figures~\ref{fig:mallows-phi} and~\ref{fig:mallows-rel}, we
visualize Mallows elections generated with $\phi \in [0,1]$ and
rel-$\phi \in [0,0.5]$ chosen uniformly at random, respectively (we
use~$\relphi \leq 0.5$ because for larger values one obtains
analogous elections, but reversed; e.g., both~$\relphi = 0$ and
$\relphi = 1$ lead to identity elections). The latter figure shows a
far more balanced distribution of points.

As computing the $\phi$ values based on $\relphi$ and the number of
candidates requires some effort, we provide sample mappings in the
table on the right side of Figure~\ref{fig:mallows-issues}.\smallskip

\definecolor{color0}{rgb}{0.12156862745098,0.466666666666667,0.705882352941177}
\definecolor{color1}{rgb}{1,0.498039215686275,0.0549019607843137}
\definecolor{color2}{rgb}{0.172549019607843,0.627450980392157,0.172549019607843}
\definecolor{color3}{rgb}{0.83921568627451,0.152941176470588,0.156862745098039}
\definecolor{color4}{rgb}{0.580392156862745,0.403921568627451,0.741176470588235}

\begin{figure}
  \centering
  \begin{subfigure}[t]{6.cm}
    \centering 
\begin{tikzpicture}[scale=0.5]

\begin{axis}[
legend cell align={left},
legend style={fill opacity=0.8, draw opacity=1, text opacity=1, at={(0.03,0.97)}, anchor=north west, draw=white!80!black, font=\LARGE},
tick align=outside,
tick pos=left,
x grid style={white!69.0196078431373!black},
xlabel={$\phi$ parameter},
label style={font=\LARGE},
xmin=0, xmax=1,
xtick style={color=black},
xtick={0,0.2,0.4,0.6,0.8,1},
xticklabels={0.0,0.2,0.4,0.6,0.8,1.0},
y grid style={white!69.0196078431373!black},
ylabel={Normalized distance from ID},
ymin=0, ymax=1,
ytick style={color=black},
ytick={0,0.2,0.4,0.6,0.8,1},
yticklabels={0.0,0.2,0.4,0.6,0.8,1.0}
]
\addplot [semithick, color0] table {%
0 0
0.01 0.010175
0.02 0.020475
0.03 0.03065
0.04 0.041
0.05 0.051175
0.06 0.0623
0.07 0.072675
0.08 0.080625
0.09 0.0928
0.1 0.103675
0.11 0.11525
0.12 0.129375
0.13 0.13695
0.14 0.15035
0.15 0.15945
0.16 0.16975
0.17 0.1805
0.18 0.196325
0.19 0.206825
0.2 0.210725
0.21 0.23265
0.22 0.23515
0.23 0.24735
0.24 0.260625
0.25 0.2729
0.26 0.285875
0.27 0.295175
0.28 0.307875
0.29 0.318775
0.3 0.333675
0.31 0.3447
0.32 0.351425
0.33 0.3679
0.34 0.372725
0.35 0.39345
0.36 0.40355
0.37 0.4094
0.38 0.424175
0.39 0.433175
0.4 0.4442
0.41 0.454825
0.42 0.469275
0.43 0.477625
0.44 0.49115
0.45 0.508275
0.46 0.5156
0.47 0.52415
0.48 0.532675
0.49 0.5488
0.5 0.55515
0.51 0.566125
0.52 0.5802
0.53 0.58795
0.54 0.599275
0.55 0.6128
0.56 0.62525
0.57 0.632125
0.58 0.6398
0.59 0.64915
0.6 0.668175
0.61 0.683125
0.62 0.6798
0.63 0.692775
0.64 0.700725
0.65 0.71355
0.66 0.7258
0.67 0.72935
0.68 0.74635
0.69 0.753075
0.7 0.756975
0.71 0.768275
0.72 0.780825
0.73 0.78595
0.74 0.79625
0.75 0.8066
0.76 0.80685
0.77 0.821075
0.78 0.824325
0.79 0.8342
0.8 0.8462
0.81 0.845325
0.82 0.857525
0.83 0.86125
0.84 0.868025
0.85 0.875875
0.86 0.88215
0.87 0.8929
0.88 0.8947
0.89 0.90025
0.9 0.900275
0.91 0.91155
0.92 0.915425
0.93 0.91825
0.94 0.92175
0.95 0.927
0.96 0.927575
0.97 0.93075
0.98 0.9309
0.99 0.93165
1 0.9316
};
\addlegendentry{m=5}
\addplot [semithick, color1]
table {%
0 0
0.01 0.00538787878787879
0.02 0.0109272727272727
0.03 0.0163878787878788
0.04 0.0222545454545455
0.05 0.027769696969697
0.06 0.0341272727272727
0.07 0.0410060606060606
0.08 0.0461212121212121
0.09 0.0542727272727273
0.1 0.058369696969697
0.11 0.0652121212121212
0.12 0.0705393939393939
0.13 0.0774424242424242
0.14 0.0844424242424242
0.15 0.091260606060606
0.16 0.0985575757575758
0.17 0.103527272727273
0.18 0.111321212121212
0.19 0.118909090909091
0.2 0.125551515151515
0.21 0.132987878787879
0.22 0.140533333333333
0.23 0.146593939393939
0.24 0.155672727272727
0.25 0.162454545454545
0.26 0.169945454545455
0.27 0.177806060606061
0.28 0.184327272727273
0.29 0.192369696969697
0.3 0.200030303030303
0.31 0.211424242424242
0.32 0.218024242424242
0.33 0.225527272727273
0.34 0.234593939393939
0.35 0.239890909090909
0.36 0.252527272727273
0.37 0.258478787878788
0.38 0.267030303030303
0.39 0.278715151515152
0.4 0.285218181818182
0.41 0.298563636363636
0.42 0.306521212121212
0.43 0.316557575757576
0.44 0.321854545454546
0.45 0.333854545454545
0.46 0.339709090909091
0.47 0.354727272727273
0.48 0.362587878787879
0.49 0.370763636363636
0.5 0.385115151515152
0.51 0.392048484848485
0.52 0.3992
0.53 0.415248484848485
0.54 0.425490909090909
0.55 0.438593939393939
0.56 0.445048484848485
0.57 0.458660606060606
0.58 0.471569696969697
0.59 0.481012121212121
0.6 0.49310303030303
0.61 0.503187878787879
0.62 0.516139393939394
0.63 0.530163636363636
0.64 0.539624242424242
0.65 0.551260606060606
0.66 0.565115151515152
0.67 0.577472727272727
0.68 0.587575757575758
0.69 0.603127272727273
0.7 0.618521212121212
0.71 0.630460606060606
0.72 0.640618181818182
0.73 0.654266666666667
0.74 0.668812121212121
0.75 0.676866666666667
0.76 0.6956
0.77 0.705248484848485
0.78 0.716121212121212
0.79 0.733163636363636
0.8 0.740848484848485
0.81 0.754981818181818
0.82 0.769909090909091
0.83 0.779242424242424
0.84 0.794593939393939
0.85 0.802624242424242
0.86 0.816369696969697
0.87 0.830157575757576
0.88 0.846042424242424
0.89 0.852048484848485
0.9 0.865812121212121
0.91 0.876036363636364
0.92 0.887551515151515
0.93 0.895945454545454
0.94 0.904381818181818
0.95 0.9104
0.96 0.918842424242424
0.97 0.92329696969697
0.98 0.928787878787878
0.99 0.929230303030303
1 0.930169696969697
};
\addlegendentry{m=10}
\addplot [semithick, color2]
table {%
0 0
0.01 0.00287368421052632
0.02 0.00574436090225564
0.03 0.00864511278195489
0.04 0.0119218045112782
0.05 0.0150932330827068
0.06 0.0181037593984962
0.07 0.0209022556390977
0.08 0.0246616541353383
0.09 0.0278135338345865
0.1 0.0314300751879699
0.11 0.034409022556391
0.12 0.0381022556390977
0.13 0.0416721804511278
0.14 0.0455157894736842
0.15 0.0482511278195489
0.16 0.0525007518796992
0.17 0.0560992481203008
0.18 0.0600195488721805
0.19 0.0636691729323308
0.2 0.0673669172932331
0.21 0.0712075187969925
0.22 0.0756255639097744
0.23 0.079981954887218
0.24 0.0837609022556391
0.25 0.0873924812030075
0.26 0.0916827067669173
0.27 0.0964736842105263
0.28 0.10064962406015
0.29 0.105410526315789
0.3 0.109881203007519
0.31 0.114363909774436
0.32 0.118553383458647
0.33 0.123222556390977
0.34 0.129633082706767
0.35 0.133864661654135
0.36 0.139384962406015
0.37 0.143711278195489
0.38 0.149476691729323
0.39 0.155317293233083
0.4 0.160574436090226
0.41 0.165948872180451
0.42 0.171645112781955
0.43 0.177938345864662
0.44 0.183187969924812
0.45 0.188863157894737
0.46 0.196126315789474
0.47 0.202727819548872
0.48 0.209947368421053
0.49 0.216458646616541
0.5 0.22223007518797
0.51 0.229777443609023
0.52 0.237554887218045
0.53 0.244189473684211
0.54 0.252697744360902
0.55 0.259993984962406
0.56 0.268833082706767
0.57 0.276932330827068
0.58 0.28543007518797
0.59 0.293178947368421
0.6 0.304291729323308
0.61 0.313209022556391
0.62 0.324028571428571
0.63 0.333342857142857
0.64 0.345401503759398
0.65 0.354813533834586
0.66 0.366771428571429
0.67 0.378774436090226
0.68 0.388888721804511
0.69 0.401959398496241
0.7 0.414354887218045
0.71 0.427126315789474
0.72 0.441657142857143
0.73 0.454968421052632
0.74 0.470562406015038
0.75 0.484735338345865
0.76 0.500813533834587
0.77 0.518192481203007
0.78 0.534812030075188
0.79 0.55216992481203
0.8 0.569882706766917
0.81 0.586424060150376
0.82 0.607703759398496
0.83 0.623366917293233
0.84 0.645254135338346
0.85 0.663798496240601
0.86 0.686081203007519
0.87 0.70595037593985
0.88 0.72469022556391
0.89 0.746048120300752
0.9 0.768270676691729
0.91 0.789751879699248
0.92 0.811584962406015
0.93 0.829604511278195
0.94 0.852234586466165
0.95 0.868998496240602
0.96 0.888839097744361
0.97 0.906505263157895
0.98 0.918669172932331
0.99 0.928004511278196
1 0.930783458646617
};
\addlegendentry{m=20}
\addplot [semithick, color3]
table {%
0 0
0.01 0.00115990396158463
0.02 0.00238031212484994
0.03 0.0036343337334934
0.04 0.00486002400960384
0.05 0.006184393757503
0.06 0.00749507803121249
0.07 0.0087968787515006
0.08 0.0102151260504202
0.09 0.0115606242496999
0.1 0.0129193277310924
0.11 0.014303481392557
0.12 0.0157102040816327
0.13 0.0172552220888355
0.14 0.0187663865546218
0.15 0.0201140456182473
0.16 0.0218050420168067
0.17 0.0231951980792317
0.18 0.0248146458583433
0.19 0.0263313325330132
0.2 0.0281171668667467
0.21 0.030041056422569
0.22 0.0314153661464586
0.23 0.0332316926770708
0.24 0.0350453781512605
0.25 0.0368794717887155
0.26 0.0387222088835534
0.27 0.040709243697479
0.28 0.0425037214885954
0.29 0.0443075630252101
0.3 0.046252581032413
0.31 0.0483877551020408
0.32 0.050358943577431
0.33 0.0526
0.34 0.0544893157262905
0.35 0.0566991596638656
0.36 0.058937575030012
0.37 0.0610941176470588
0.38 0.0636509003601441
0.39 0.0661582232893157
0.4 0.068480912364946
0.41 0.0709822328931573
0.42 0.0736917166866747
0.43 0.0761483793517407
0.44 0.0788888355342137
0.45 0.0817226890756303
0.46 0.0844727490996399
0.47 0.0873517406962785
0.48 0.0904686674669868
0.49 0.0935968787515006
0.5 0.0968621848739496
0.51 0.100165906362545
0.52 0.10364849939976
0.53 0.107207202881152
0.54 0.110868427370948
0.55 0.114397118847539
0.56 0.118658823529412
0.57 0.12243481392557
0.58 0.126810324129652
0.59 0.131446338535414
0.6 0.136403841536615
0.61 0.141014885954382
0.62 0.146368067226891
0.63 0.151313805522209
0.64 0.156498919567827
0.65 0.162744057623049
0.66 0.16889843937575
0.67 0.174956542617047
0.68 0.18149699879952
0.69 0.188504201680672
0.7 0.195414165666267
0.71 0.203357262905162
0.72 0.212523649459784
0.73 0.22167274909964
0.74 0.230198799519808
0.75 0.240183673469388
0.76 0.251229771908763
0.77 0.263235774309724
0.78 0.275045618247299
0.79 0.289174069627851
0.8 0.302064825930372
0.81 0.317788475390156
0.82 0.33484393757503
0.83 0.352229531812725
0.84 0.370998799519808
0.85 0.392236974789916
0.86 0.415296518607443
0.87 0.439588475390156
0.88 0.466898199279712
0.89 0.497748619447779
0.9 0.528759423769508
0.91 0.564878511404562
0.92 0.603329171668667
0.93 0.643175990396158
0.94 0.686746218487395
0.95 0.732585354141657
0.96 0.780596878751501
0.97 0.828399039615846
0.98 0.874944057623049
0.99 0.913270828331332
1 0.929997118847539
};
\addlegendentry{m=50}
\addplot [semithick, color4]
table {%
0 0
0.01 0.000592679267926793
0.02 0.00121344134413441
0.03 0.00184368436843684
0.04 0.00247434743474347
0.05 0.00310579057905791
0.06 0.00377803780378038
0.07 0.00445070507050705
0.08 0.00511785178517852
0.09 0.00581362136213621
0.1 0.00652187218721872
0.11 0.00729012901290129
0.12 0.0079979597959796
0.13 0.00874011401140114
0.14 0.00948346834683469
0.15 0.0102541254125413
0.16 0.0110244224422442
0.17 0.0117794779477948
0.18 0.0126163816381638
0.19 0.0134466246624662
0.2 0.0143495349534954
0.21 0.0151075307530753
0.22 0.0160439243924392
0.23 0.016950795079508
0.24 0.0177590759075908
0.25 0.0187059705970597
0.26 0.0196276627662766
0.27 0.020626402640264
0.28 0.0216479447944794
0.29 0.0225363936393639
0.3 0.0235730573057306
0.31 0.0245314131413141
0.32 0.0256975697569757
0.33 0.0266729072907291
0.34 0.0278117611761176
0.35 0.0289525352535253
0.36 0.0301005100510051
0.37 0.0312787878787879
0.38 0.0325251725172517
0.39 0.0337182718271827
0.4 0.0349533153315331
0.41 0.0363224122412241
0.42 0.0375917791779178
0.43 0.0389144914491449
0.44 0.0403288328832883
0.45 0.0416700270027003
0.46 0.0433204320432043
0.47 0.0448558655865587
0.48 0.046379297929793
0.49 0.0479362136213621
0.5 0.0497027902790279
0.51 0.0514538853885389
0.52 0.0532976297629763
0.53 0.0551774377437744
0.54 0.0569743774377438
0.55 0.0590064206420642
0.56 0.0611601560156016
0.57 0.0632694869486949
0.58 0.065637503750375
0.59 0.068029402940294
0.6 0.0704244224422442
0.61 0.073009900990099
0.62 0.0754980498049805
0.63 0.0784837683768377
0.64 0.0816364836483648
0.65 0.0847875187518752
0.66 0.0876471647164716
0.67 0.0914298229822983
0.68 0.0950934293429343
0.69 0.0991521152115212
0.7 0.103062406240624
0.71 0.107641464146415
0.72 0.11209702970297
0.73 0.117270687068707
0.74 0.122414701470147
0.75 0.128199759975998
0.76 0.134408700870087
0.77 0.141050165016502
0.78 0.147962076207621
0.79 0.15608796879688
0.8 0.164411461146115
0.81 0.173604260426043
0.82 0.183998859885989
0.83 0.194975757575758
0.84 0.207641644164416
0.85 0.221317671767177
0.86 0.236731353135314
0.87 0.25506798679868
0.88 0.274460486048605
0.89 0.297298769876988
0.9 0.3241099909991
0.91 0.35437101710171
0.92 0.390027782778278
0.93 0.432061446144614
0.94 0.482864326432643
0.95 0.541804200420042
0.96 0.613138133813381
0.97 0.692549534953495
0.98 0.784302550255026
0.99 0.874736453645365
1 0.929670327032703
};
\addlegendentry{m=100}
\end{axis}

\end{tikzpicture}

\begin{tikzpicture}[scale=0.5]

\definecolor{color0}{rgb}{0.12156862745098,0.466666666666667,0.705882352941177}
\definecolor{color1}{rgb}{1,0.498039215686275,0.0549019607843137}
\definecolor{color2}{rgb}{0.172549019607843,0.627450980392157,0.172549019607843}
\definecolor{color3}{rgb}{0.83921568627451,0.152941176470588,0.156862745098039}
\definecolor{color4}{rgb}{0.580392156862745,0.403921568627451,0.741176470588235}

\begin{axis}[
legend cell align={left},
legend style={fill opacity=0.8, draw opacity=1, text opacity=1, at={(0.03,0.97)}, anchor=north west, draw=white!80!black, font=\LARGE},
tick align=outside,
tick pos=left,
x grid style={white!69.0196078431373!black},
xlabel={$\phi$ parameter},
label style={font=\LARGE},
xmin=0, xmax=1,
xtick style={color=black},
xtick={0,0.2,0.4,0.6,0.8,1},
xticklabels={0.0,0.2,0.4,0.6,0.8,1.0},
y grid style={white!69.0196078431373!black},
ylabel={Exp. rel. number of swaps},
ymin=0, ymax=0.5,
ytick style={color=black},
ytick={0,0.1,0.2,0.3,0.4,0.5},
yticklabels={0,0.1,0.2,0.3,0.4,0.5}
]
\addplot [semithick, color0]
table {%
0 0
0.01 0.00402009798990398
0.02 0.00808076767410693
0.03 0.0121825355057186
0.04 0.0163258774146485
0.05 0.0205112174923554
0.06 0.0247389266613915
0.07 0.0290093213472365
0.08 0.0333226621691757
0.09 0.0376791526662478
0.1 0.0420789380735416
0.11 0.0465221041633761
0.12 0.0510086761651335
0.13 0.0555386177767342
0.14 0.060111830279943
0.15 0.0647281517708735
0.16 0.0693873565162099
0.17 0.0740891544447898
0.18 0.0788331907832966
0.19 0.083619045843876
0.2 0.0884462349705421
0.21 0.0933142086502556
0.22 0.0982223527935529
0.23 0.103169989188578
0.24 0.108156376131326
0.25 0.11318070923384
0.26 0.118242122411048
0.27 0.123339689045826
0.28 0.128472423330817
0.29 0.133639281784468
0.3 0.138839164937668
0.31 0.144070919186362
0.32 0.149333338804483
0.33 0.154625168110605
0.34 0.159945103780752
0.35 0.165291797298943
0.36 0.170663857536241
0.37 0.176059853448294
0.38 0.181478316880689
0.39 0.186917745470837
0.4 0.192376605634574
0.41 0.197853335625224
0.42 0.20334634865254
0.43 0.208854036048676
0.44 0.214374770468169
0.45 0.219906909108891
0.46 0.225448796940903
0.47 0.230998769930315
0.48 0.236555158245447
0.49 0.242116289432892
0.5 0.247680491551459
0.51 0.253246096252448
0.52 0.25881144179522
0.53 0.264374875987644
0.54 0.269934759041632
0.55 0.27548946633469
0.56 0.281037391069159
0.57 0.286576946821585
0.58 0.292106569975475
0.59 0.297624722031518
0.6 0.303129891790165
0.61 0.308620597402325
0.62 0.314095388284738
0.63 0.319552846897423
0.64 0.324991590381387
0.65 0.330410272055572
0.66 0.335807582772746
0.67 0.341182252134769
0.68 0.346533049568352
0.69 0.351858785263027
0.7 0.357158310973682
0.71 0.362430520690515
0.72 0.36767435117979
0.73 0.372888782399205
0.74 0.378072837792094
0.75 0.383225584465021
0.76 0.388346133253638
0.77 0.393433638681913
0.78 0.398487298820056
0.79 0.403506355046635
0.8 0.408490091720465
0.81 0.413437835767966
0.82 0.418348956191706
0.83 0.423222863505843
0.84 0.428059009104181
0.85 0.432856884566447
0.86 0.437616020908371
0.87 0.442335987780982
0.88 0.447016392624432
0.89 0.451656879781507
0.9 0.456257129575803
0.91 0.460816857359364
0.92 0.465335812534391
0.93 0.469813777553425
0.94 0.474250566902184
0.95 0.478646026069038
0.96 0.483000030504858
0.97 0.487312484576778
0.98 0.491583320519178
0.99 0.495812497384961
1 0.5
};
\addlegendentry{m=5}
\addplot [semithick, color1]
table {%
0 0
0.01 0.00201568956440043
0.02 0.00406329983720956
0.03 0.00614365612849768
0.04 0.00825760181560101
0.05 0.0104059993357909
0.06 0.0125897311513812
0.07 0.0148097006877446
0.08 0.0170668332441166
0.09 0.0193620768764354
0.1 0.0216964032507756
0.11 0.0240708084652045
0.12 0.0264863138370984
0.13 0.0289439666521165
0.14 0.0314448408701331
0.15 0.0339900377824827
0.16 0.0365806866138667
0.17 0.0392179450612227
0.18 0.0419029997607557
0.19 0.0446370666731924
0.2 0.0474213913761347
0.21 0.0502572492511862
0.22 0.0531459455522852
0.23 0.0560888153404374
0.24 0.0590872232687892
0.25 0.0621425632007489
0.26 0.0652562576426508
0.27 0.06842975697129
0.28 0.0716645384355529
0.29 0.0749621049103461
0.3 0.0783239833801193
0.31 0.0817517231284972
0.32 0.085246893609931
0.33 0.0888110819788644
0.34 0.0924458902517252
0.35 0.0961529320771463
0.36 0.0999338290902118
0.37 0.103790206827268
0.38 0.107723690178983
0.39 0.111735898360894
0.4 0.115828439382764
0.41 0.120002904000613
0.42 0.124260859138447
0.43 0.128603840770461
0.44 0.133033346258857
0.45 0.137550826147475
0.46 0.142157675417184
0.47 0.146855224215403
0.48 0.151644728079249
0.49 0.156527357679635
0.5 0.161504188122074
0.51 0.166576187849015
0.52 0.171744207198081
0.53 0.177008966680576
0.54 0.182371045054906
0.55 0.187830867279968
0.56 0.193388692443983
0.57 0.199044601774373
0.58 0.204798486843999
0.59 0.210650038098044
0.6 0.216598733833802
0.61 0.22264382977238
0.62 0.228784349366465
0.63 0.235019074991607
0.64 0.241346540169671
0.65 0.247765022971894
0.66 0.25427254074517
0.67 0.260866846298567
0.68 0.267545425677557
0.69 0.274305497640917
0.7 0.281144014939784
0.71 0.288057667479991
0.72 0.295042887427793
0.73 0.302095856295639
0.74 0.309212514019215
0.75 0.316388570009917
0.76 0.323619516138841
0.77 0.330900641579827
0.78 0.338227049410753
0.79 0.345593674844764
0.8 0.352995304937153
0.81 0.360426599589796
0.82 0.367882113654076
0.83 0.375356319915509
0.84 0.382843632729456
0.85 0.390338432067552
0.86 0.397835087729181
0.87 0.405327983471493
0.88 0.412811540815139
0.89 0.420280242290917
0.9 0.427728653904645
0.91 0.43515144661337
0.92 0.442543416625091
0.93 0.44989950435594
0.94 0.45721481190261
0.95 0.464484618913168
0.96 0.4717043967656
0.97 0.478869820989807
0.98 0.485976781894886
0.99 0.493021393388628
1 0.5
};
\addlegendentry{m=10}
\addplot [semithick, color2]
table {%
0 0
0.01 0.00100903227056906
0.02 0.0020364743438289
0.03 0.00308285724367295
0.04 0.00414873025457214
0.05 0.00523466189254851
0.06 0.00634124091659029
0.07 0.00746907738415349
0.08 0.00861880375455895
0.09 0.0097910760442753
0.1 0.0109865750382766
0.11 0.012206007561881
0.12 0.0134501078177133
0.13 0.0147196387926933
0.14 0.0160153937402308
0.15 0.017338197743112
0.16 0.0186889093628927
0.17 0.0200684223819634
0.18 0.0214776676448387
0.19 0.0229176150056277
0.2 0.0243892753890848
0.21 0.025893702973107
0.22 0.0274319975010432
0.23 0.0290053067327111
0.24 0.0306148290435741
0.25 0.0322618161821221
0.26 0.0339475761961107
0.27 0.035673476538954
0.28 0.0374409473682216
0.29 0.0392514850488601
0.3 0.0411066558744337
0.31 0.0430081000203428
0.32 0.0449575357436224
0.33 0.0469567638445275
0.34 0.0490076724056425
0.35 0.0511122418246955
0.36 0.0532725501575603
0.37 0.0554907787880472
0.38 0.0577692184409665
0.39 0.0601102755545103
0.4 0.0625164790271637
0.41 0.0649904873530103
0.42 0.0675350961573225
0.43 0.0701532461415573
0.44 0.0728480314431533
0.45 0.0756227084106037
0.46 0.0784807047879538
0.47 0.0814256292948017
0.48 0.0844612815777811
0.49 0.087591662496958
0.5 0.0908209846951682
0.51 0.0941536833795621
0.52 0.0975944272219852
0.53 0.101148129257703
0.54 0.104819957629782
0.55 0.108615345988478
0.56 0.112540003310647
0.57 0.116599922852747
0.58 0.120801389892038
0.59 0.125150987843388
0.6 0.129655602263653
0.61 0.134322422171684
0.62 0.139158938020138
0.63 0.14417293555631
0.64 0.149372484704775
0.65 0.15476592249721
0.66 0.160361828968094
0.67 0.16616899483392
0.68 0.172196379684832
0.69 0.178453059349578
0.7 0.184948161057945
0.71 0.191690785031973
0.72 0.198689911203004
0.73 0.205954289892666
0.74 0.213492315529947
0.75 0.221311882821939
0.76 0.229420225269705
0.77 0.237823736536947
0.78 0.246527775946411
0.79 0.255536460296588
0.8 0.264852445247102
0.81 0.274476700686811
0.82 0.284408285727286
0.83 0.294644130188175
0.84 0.305178830571289
0.85 0.316004469449835
0.86 0.327110467808153
0.87 0.338483480032794
0.88 0.350107340865335
0.89 0.361963072594969
0.9 0.374028959052202
0.91 0.386280690580585
0.92 0.398691581197698
0.93 0.411232855768969
0.94 0.423874001433385
0.95 0.436583174011794
0.96 0.449327646989364
0.97 0.462074288173124
0.98 0.47479004751344
0.99 0.487442438995138
1 0.5
};
\addlegendentry{m=20}
\addplot [semithick, color3]
table {%
0 0
0.01 0.000403874640357897
0.02 0.000815653080234022
0.03 0.00123557381864837
0.04 0.00166388469254588
0.05 0.00210084336980022
0.06 0.00254671787034566
0.07 0.00300178711764751
0.08 0.00346634152288989
0.09 0.00394068360444512
0.1 0.00442512864539256
0.11 0.00492000539207895
0.12 0.00542565679695887
0.13 0.00594244080922562
0.14 0.00647073121704196
0.15 0.00701091854551029
0.16 0.00756341101488598
0.17 0.00812863556393964
0.18 0.00870703894381868
0.19 0.00929908888825002
0.2 0.00990527536647025
0.21 0.0105261119258726
0.22 0.0111621371320296
0.23 0.0118139161144922
0.24 0.012482042227594
0.25 0.0131671388364055
0.26 0.0138698612390095
0.27 0.0145908987374095
0.28 0.0153309768706578
0.29 0.0160908598252157
0.3 0.0168713530391551
0.31 0.0176733060185983
0.32 0.0184976153868008
0.33 0.0193452281885405
0.34 0.0202171454750156
0.35 0.0211144261973167
0.36 0.022038191439774
0.37 0.0229896290281317
0.38 0.0239699985516444
0.39 0.0249806368428818
0.4 0.0260229639643614
0.41 0.0270984897571978
0.42 0.0282088210138713
0.43 0.0293556693451175
0.44 0.0305408598199703
0.45 0.0317663404683392
0.46 0.0330341927473822
0.47 0.0343466430865947
0.48 0.035706075642276
0.49 0.0371150464101985
0.5 0.0385762988663188
0.51 0.0400927813297096
0.52 0.0416676662701533
0.53 0.0433043718157085
0.54 0.0450065857538598
0.55 0.0467782923645816
0.56 0.0486238024759392
0.57 0.0505477871941237
0.58 0.0525553158317151
0.59 0.0546518986424549
0.6 0.0568435350702395
0.61 0.059136768337144
0.62 0.061538747333319
0.63 0.0640572969343333
0.64 0.066700998063303
0.65 0.0694792790408818
0.66 0.0724025200313156
0.67 0.075482172702984
0.68 0.0787308975828601
0.69 0.0821627220008546
0.7 0.0857932219947775
0.71 0.0896397320779495
0.72 0.0937215873497566
0.73 0.0980604030315326
0.74 0.102680397090499
0.75 0.10760876209146
0.76 0.112876092651422
0.77 0.118516874641363
0.78 0.124570041226761
0.79 0.131079598417666
0.8 0.138095318190491
0.81 0.145673489251674
0.82 0.153877702448799
0.83 0.162779627408483
0.84 0.172459706264591
0.85 0.183007645918538
0.86 0.194522528795617
0.87 0.207112281563606
0.88 0.220892143782509
0.89 0.23598167441299
0.9 0.252499748997004
0.91 0.270556982417281
0.92 0.290245137824398
0.93 0.311623450364934
0.94 0.334702495502595
0.95 0.359427280463487
0.96 0.385662470814745
0.97 0.413183659205362
0.98 0.441678698615934
0.99 0.470761772968971
1 0.500000000000001
};
\addlegendentry{m=50}
\addplot [semithick, color4]
table {%
0 0
0.01 0.000201979179694733
0.02 0.000407996603343835
0.03 0.000618175693014459
0.04 0.000832644864990311
0.05 0.00105153779300479
0.06 0.00127499368744961
0.07 0.00150315759178449
0.08 0.00173618069747407
0.09 0.00197422067888893
0.1 0.00221744204972959
0.11 0.00246601654266577
0.12 0.0027201235140316
0.13 0.00297995037557929
0.14 0.00324569305547421
0.15 0.00351755649091208
0.16 0.00379575515495855
0.17 0.00408051362045409
0.18 0.00437206716409605
0.19 0.00467066241410802
0.2 0.00497655804523759
0.21 0.00529002552519198
0.22 0.00561134991703037
0.23 0.00594083074248847
0.24 0.00627878291171984
0.25 0.00662553772550775
0.26 0.00698144395663721
0.27 0.00734686901782994
0.28 0.00772220022444449
0.29 0.00810784616104178
0.3 0.00850423816192656
0.31 0.00891183191691415
0.32 0.00933110921485593
0.33 0.00976257983890988
0.34 0.0102067836291858
0.35 0.0106642927302606
0.36 0.0111357140431764
0.37 0.0116216919039476
0.38 0.0121229110133509
0.39 0.012640099645916
0.4 0.0131740331696315
0.41 0.0137255379120064
0.42 0.014295495412872
0.43 0.0148848471097743
0.44 0.0154945995081167
0.45 0.0161258298955165
0.46 0.0167796926683102
0.47 0.0174574263479976
0.48 0.0181603613768958
0.49 0.0188899287957107
0.5 0.0196476699214627
0.51 0.0204352471627121
0.52 0.0212544561308291
0.53 0.0221072392318477
0.54 0.022995700954018
0.55 0.0239221251025691
0.56 0.0248889942766408
0.57 0.025899011935394
0.58 0.0269551274629016
0.59 0.0280605647169574
0.6 0.0292188546384661
0.61 0.0304338726094164
0.62 0.0317098813834667
0.63 0.033051580580076
0.64 0.0344641639388529
0.65 0.0359533857856183
0.66 0.0375256384788408
0.67 0.0391880430019443
0.68 0.0409485553661534
0.69 0.0428160921200021
0.7 0.0448006790650867
0.71 0.0469136283062251
0.72 0.0491677500892951
0.73 0.0515776075984273
0.74 0.0541598251278397
0.75 0.0569334629935371
0.76 0.0599204764566916
0.77 0.0631462811418123
0.78 0.0666404544355221
0.79 0.0704376118291384
0.8 0.0745785100817919
0.81 0.0791114467813482
0.82 0.0840940502528083
0.83 0.0895955873931004
0.84 0.0956999633177597
0.85 0.102509649884432
0.86 0.110150864484947
0.87 0.118780427900675
0.88 0.128594853351432
0.89 0.139842325134287
0.9 0.152838216650616
0.91 0.167984427601099
0.92 0.185791529524223
0.93 0.206899360325231
0.94 0.232084339700065
0.95 0.262228061755775
0.96 0.298202297199608
0.97 0.340614060292097
0.98 0.389392508623967
0.99 0.443343006455033
1 0.499999999999999
};
\addlegendentry{m=100}
\end{axis}

\end{tikzpicture}

  \end{subfigure}
  \quad
  \begin{subfigure}[t]{8cm}
    \small
    \centering
    \begin{tabular}{l|ccccc}
                        \toprule
                        $\relphi \backslash m$ & 5 & 10 & 20 & 50 & 100 \\
                        \midrule
                        0 & 0.000 & 0.000 & 0.000 & 0.000 & 0.000 \\
                        0.05 & 0.118 & 0.209 & 0.345 & 0.567 & 0.724 \\
                        0.1 & 0.224 & 0.360 & 0.527 & 0.734 & 0.846 \\
                        0.15 & 0.321 & 0.477 & 0.641 & 0.815 & 0.898 \\
                        0.2 & 0.414 & 0.572 & 0.722 & 0.864 & 0.927 \\
                        0.25 & 0.504 & 0.653 & 0.784 & 0.899 & 0.946 \\
                        0.3 & 0.594 & 0.727 & 0.835 & 0.925 & 0.960 \\
                        0.35 & 0.687 & 0.796 & 0.880 & 0.946 & 0.972 \\
                        0.4 & 0.783 & 0.863 & 0.921 & 0.965 & 0.982 \\
                        0.45 & 0.886 & 0.930 & 0.961 & 0.983 & 0.991 \\
                        0.5 & 1.000 & 1.000 & 1.000 & 1.000 & 1.000 \\
                        \bottomrule
    \end{tabular}
  \end{subfigure}
  \caption{\label{fig:mallows-issues}Average normalized positionwise
    distances of Mallows elections from $\ID$ (plot on the top-left),
    relative expected number of swaps in votes drawn from the Mallows
    model (plot on the bottom-left), both depending on $\phi$ and for
    different numbers $m$ of candidates, and---in the table---the
    values of $\phi$ such that $\relswaps(m,\phi) = \relphi$ for
    $m\in \{5,10,20,50,100\}$ and
    $\relphi\in \{0,0.05,0.1,0.15,0.2,0.25,0.3,0.35,0.4,0.45,0.5\}$.
  }
\end{figure}

\noindent\textbf{Importance of the New Normalization.} 
The new parameterization of Mallows seems to be important. The Mallows
model is often used in experiments and---in light of our
findings---using a fixed $\phi$ for different numbers of candidates or
drawing $\phi$ from a distribution independent of the number of
candidates, may be questionable. Yet, this is not uncommon, as
witnessed, e.g., in the works of
\citet{bac-lev-lew-zic:c:misrepresentation},
\citet{bet-bre-nie:j:kemeny},
\citet{gar-gel-sak-goe:c:top-three},
\citet{gol-lan-mat-per:c:rank-dependent},
\citet{sko-fal-sli:j:multiwinner},
and in a number of other
papers. We mention these works as examples only; their authors
designed their experiments as best practice suggested at the time and
we do not challenge their high-level conclusions. Our point is that
given the current evidence,
they might 
prefer to design their
experiments a bit differently.

\subsection{Real-Life Elections on the Map} \label{sub:maprel}
Let us now consider where 
real-life elections appear on the map. We start by describing the
datasets that we use (mostly from
PrefLib~\cite{mat-wal:c:preflib}).
Whenever we speak of real-life elections in this section, we mean
elections from our datasets.

We chose eleven real-life datasets, where each belongs to one of three groups.
The first group contains \textit{political} elections: city council
elections from Glasgow and Aspen~\cite{openstv}, elections from
North Dublin, Meath (Irish), and elections held by
non-profit organizations, trade unions, and professional organizations
(ERS). The second group consists of \textit{sport} elections: Tour de
France
(TDF), 
Giro d'Italia (GDI), 
speed skating, 
and figure skating (the former three dataset are due to us). The last
group consists of \textit{surveys}: preferences over Sushi, T-Shirt
designs, and costs of living and population in different cities~\cite{caragiannis2019optimizing}. For TDF and GDI, each race is a
vote, and each season is an election. For speed skating, each lap is a
vote, and each competition is an election. For figure skating, each
judge's opinion is a vote, and each competition is an election.

\paragraph{Preprocessing the Data.} 
Each of our datasets consists of a number of preference profiles,
where each profile consists of (possibly) partial preference
orders. Since for our map we need elections with complete preference
orders, we preprocess the data as follows.\footnote{We speak of
  \emph{elections} when we mean the collections of preference orders
  that we used in the map. We speak of \emph{preference profiles} when
  we mean collections of 
  preference orders 
   at various stages of
  preprocessing.}  First, for our new
sport-related datasets (i.e., for TDG, GDI, and speed skating) we
delete candidates and voters until each remaining candidate is ranked
by at least 70\% of the voters, and each voter ranks at least 70\% of
the candidates.  Second, for all the datasets we disregard those
profiles that contain fewer than ten candidates.  Third, we extend
each partial preference order in each remaining profile as follows
(our approach is a much-simplified variant of the technique introduced
by \citet{dou:thesis:partial-social-choice}):
\begin{enumerate}
\item If some $t$ top candidates are ranked and the remaining ones are
  not (except that they are reported to be below the top ones), then
  we fill~$v$ iteratively: (1) We draw uniformly at random one of
  the original votes from the same profile that ranks at least the top
  $t+1$ candidates and that agrees with $v$ on the top $t$ positions,
  and (2) we extend $v$ with whoever is ranked on the $(t+1)$-st
  position in this drawn vote (if there are no votes to sample from,
  then we extend $v$ with a candidate chosen uniformly at random). We
  repeat this process until $v$ is complete.
\item If a vote contains a tie of a different type than described
  above, then we break it uniformly at random.
\end{enumerate}
After applying these preprocessing steps, each dataset contains
profiles with ten or more candidates, and with complete votes. For
each such profile, we select the ten candidates with the highest Borda
score and remove the other ones.  Finally, we delete some of the
profiles based on the number of voters they contain (as compared to
the other profiles in a given dataset; see \Cref{ap:datasets} for details).  
We refer to the resulting datasets as \emph{intermediate}.

We treat each of the intermediate real-life datasets as a separate
election model, from which we sample fifteen elections. 
We sample each election as follows. 
First, we randomly select
one of the profiles. Second, we sample $100$ votes from the profile
uniformly at random (this implies that for profiles with less than
$100$ votes, we select some of the votes multiple times, and for
profiles with more than $100$ votes, we do not select some votes at
all).  After executing this procedure, we arrive at $11$ datasets,
each containing $15$ elections consisting of $100$ complete and strict
votes over $10$ candidates, which we use for our experiments.
For a more detailed description of our data and its processing, see \Cref{ap:datasets}.

\paragraph{The Map.}
In Figure~\ref{fig:10x100RW}, we show a map of our real-life elections
along with the compass, Mallows, and Urn elections. For readability we
present Mallows and Urn elections as large, pale-colored areas. Not
all real-life elections form clear clusters, hence the labels refer to
the largest compact groupings.

While the map is not a perfect representation of distances among
elections,
analyzing it nevertheless leads to many conclusions.
Most strikingly, real-life elections occupy a very
limited area of the map; this is especially true for political
elections and surveys. Except for several sport elections, all
elections are closer to $\UN$ than to $\ID$, and none of the
real-life elections falls in the top-right part of the map. Another
observation is that Mallows elections go right through the real-life
elections, while Urn elections are on average far away. This means
that for most real-life elections there exists a parameter~$\phi$ such
that elections generated according to the Mallows model with that
parameter are relatively close
(see the next section for specific recommendations).
    
Most of the political elections lie close to one another and are
located next to Mallows elections and high-dimensional hypercube
ones. At the same time, sport elections are spread over a larger part
of the map and, with the exception of GDI, are shifted toward
$\ID$. As to the surveys, the City survey is basically equivalent to a
sample from IC. The Sushi survey is surprisingly similar to political
elections. The T-Shirt survey is shifted toward stratification
(apparently, people often agree which designs are better and which are
worse).

\subsection{Capturing Real-Life Elections} \label{sub:recom}
In this section we analyze how to choose the $\relphi$ parameter so
that elections generated using the Mallows model with our normalization resemble our real-life
ones. We consider four different datasets each consisting of elections with 10 candidates and 100 voters (created as descibed in \Cref{sub:maprel}): the set of all political
elections, the set of all sport elections, the set of all survey
elections, and the combined set of all our real-life elections. For
each of these four datasets, to find the value of $\relphi$ that
produces elections that are as similar as possible to the respective
real-life elections, we conduct the following experiment. For each
$\relphi\in \{0,0.001,0.002,...,0.499,0.5\}$, we generate $100$
elections with 10 candidates and 100 voters from the Mallows model
with the given $\relphi$ parameter. Subsequently, we compute the
average distance between these elections and the elections from the
respective dataset. Finally, we select the value of $\relphi$ that minimizes this
distance. We present the results of this experiment in
\Cref{ta:realphiRL}.

\begin{table*}
\centering
\begin{tabular}{l|c|c|c|c}
    \toprule
    Type of elections& Value of & Average Normalized & Norm. Std.& Num. of \\
   & $\relphi$ & Distance & Dev. &elections \\
    \midrule
    Political elections & $0.375$ & $0.15$ & $0.036$ & 60 \\
    Sport elections & $0.267$ & $0.27$ & $0.080$ & 60 \\
    Survey elections & $0.365$ & $0.20$ & $0.034$ & 45 \\ 
   All real-life elections & 0.350 &  0.22& 0.106& 165\\
    \bottomrule
  \end{tabular}
  \caption{\label{ta:realphiRL}Values of $\relphi$ such that elections
    generated with Mallows model for $m=10$ are, on average, as close
    as possible to elections from the respective dataset. We include
    the average distance of elections generated with Mallows model for
    this parameter $\relphi$ from the elections from the dataset as
    well as the standard deviation, both normalized by $D(10)=33$.
    The last column gives the number of elections in the respective
    real-life dataset.}
\end{table*}

Recall that in the previous section we have observed that a majority
of real-life elections are close to some elections generated from the
Mallows model with a certain dispersion parameter. However, we have
also seen that the real-life datasets consist of elections that differ
to a certain extent from one another (in particular, this is very
visible for the sports elections).  Thus, it is to be expected that
elections drawn from the Mallows model for a fixed dispersion
parameter are at some non-zero (average) distance from the real-life
ones. Indeed, this is the case here.  Nevertheless, the more
homogeneous political elections and survey elections can be quite well
captured using the Mallows model with parameter $\relphi=0.375$ and
$\relphi=0.365$, respectively.  Generally speaking, if one wants to
generate elections that should be particularly close to elections from
the real world, then choosing a $\relphi$ value between $0.35$ and
$0.39$ is a good strategy. If, however, one wants to capture the full
spectrum of real-life elections, then we recommend using the Mallows
model with different values of $\relphi$ from the interval
$[0.25,0.4]$. In \Cref{ssub:malm}, we provide some evidence that these
recommendations are also applicable for different numbers of
candidates.

\section{Conclusions and Future Work}
Perhaps the most important conclusion from our work is that the
Mallows model is very good at generating elections similar, under positionwise distance, to those
that appear in the real world, but to achieve this effect one needs to
choose its dispersion parameter carefully. Specifically, the parameter
should depend both on the type of elections that we are interested in
(such as, e.g., political ones or sports ones) and on the number of
candidates in the election. We have provided some recommendations for
its choice.

Our work leads to a number of open problems. On the practical side,
all our experiments regarded elections with exactly $10$
candidates. It is important to extend them to different numbers of
candidates and validate that our conclusions still hold (see
\Cref{ssub:malm} for an initial discussion regarding this issue). On
the mathematical side, one of the most intriguing question is whether,
indeed, the identity and uniformity election are the two farthest
ones in the positionwise metric.\bigskip

\noindent\textbf{Acknowledgments.} NB was supported by the DFG project
MaMu (NI 369/19). PF conducted this research based on his ERC
project PRAGMA.

\bibliographystyle{plainnat}

\newpage 

\appendix

{\centering
  \huge\bf
  Appendix

\bigskip
}

\section{Missing Material from \Cref{se:recov}} \label{se:recov_app}
\subsection{Missing Proofs and Discussion}

\polyalgo*
\begin{proof}
  Let $X$ be our input $m \times m$ matrix and let
  $C = \{c_1, \ldots, c_m\}$ be a set of candidates.  Our algorithm
  creates an election $E=(C, V)$ iteratively, as follows. In each
  iteration we first create a bipartite graph $G$ with vertex sets
  $A=\{a_1,\dots ,a_m\}$ and $B=\{b_1,\dots ,b_m\}$. For each
  $i, j \in [m]$, if $x_{i,j}$ is nonzero, then we put an edge
  between~$a_i$ and~$b_j$ (vertices in $A$ correspond to rows of $X$
  and vertices in $B$ correspond to the columns).  Next, we compute a
  perfect matching~$M$ in~$G$ (we will see later that it is guaranteed
  to exist).  Let $v$ be the vote that ranks $c_j$ on position $i$
  exactly if $M(a_i)=b_j$ ($v$ is well-defined because $M$ is a
  perfect matching).  Let $P$ be the position matrix corresponding to
  vote $v$, i.e., to election $(C,(v))$, and let $z$ be the largest
  integer such that $X-zP$ contains only non-negative entries.  Then,
  we add $z$ copies of $v$ to $V$ and set $X:=X-zP$.  We proceed to
  the next iteration until $X$ becomes the zero matrix.
  
  To prove the correctness of the algorithm, we show that at each
  iteration the constructed graph $G$ has a perfect matching.
  Let us assume that
  this is not the case.  Note that each row and each column in the
  current $X$ sums up to the same integer, say $n'$.  Since there is
  no perfect matching, by Hall's theorem, there is a subset of
  vertices $A'\subseteq A$ such that the neighborhood $B'\subseteq B$
  of $A'$ in $G$ contains fewer than $|A'|$ vertices. Yet, we have
  that $\sum_{a_i\in A', b_j\in B'} x_{i,j}=n' |A'|$, as we sum up all
  the nonzero entries of each row corresponding to a vertex from $A'$.
  However, this implies that $|B'| \geq |A'|$ because each column
  corresponding to a vertex from $B'$ sums up to $n'$, but we do not
  necessarily include all its nonzero entries. This is a
  contradiction.
  
  The algorithm terminates after at most $m^2-m+1$ steps (in each step
  at least one more entry of $X$ becomes zero, and in the last step,
  $m$ entries become zero).  Each step requires $\mathcal{O}(m^{2.5})$
  time to compute the matching, so the overall running time is
  $\mathcal{O}(m^{4.5})$.  This implies that $V$ contains at most
  $m^2-m+1$ different votes; indeed, using a similar argument as in  \citet{leep1999marriage} it can be shown that the algorithm always terminates
  after at most $m^2-2m+2$ steps.
\end{proof}

The above proof shows
that if all the entries of a given position matrix are at least
$t > 0$, then we can create an election that induces this matrix by
first choosing $t$ votes completely arbitrarily, and only then
resorting to matching in a bipartite graph.
In consequence, position matrices where all entries are large
correspond to large and varied sets of elections. While in some
situations this is unavoidable (e.g., when one considers impartial
culture elections with many more voters than candidates), it may mean
that the features observed for one election corresponding to a given
matrix
are not shared by many of the others.

Fortunately, the matrices in the datasets of
\citet{szu-fal-sko-sli-tal:c:map} either contain zero entries or have a
very small smallest entry (in the $10 \times 100$ dataset, 
69\% of the matrices have zero entries, and among the remaining ones,
the average value of the smallest entry is 2.61 with standard
deviation 1.15).

Naturally, if a matrix does have some zero entries it still may
correspond to a large and varied set of elections; we simply claim
that if all the entries are large then this is, in essence,
unavoidable.

\freqtopos*
\begin{proof}[Proof (continued from \Cref{se:recov}).]
  Here we provide a fully deterministic algorithm for computing the
  matrix $D$ from the proof presented in the main body, which does not
  invoke dependent rounding, and which minimizes the value
  $\sum_{1 \leq i,j \leq m} |nx_{i,j}-p_{i,j}|$.
  
  Consider matrix~$Y$ from the first part of the proof and let
  $B = \sum_{1 \leq i,j \leq m} y_{i,j}$.

  We form a flow network with source~$s$, nodes $v_{i,j}$ for each
  $i,j \in [m]$, ``pre-sink'' nodes $t_1, \ldots, t_m$, and sink
  node~$t$. For each $i \in [m]$, we have a directed path which starts
  at the source node~$s$, then goes to $v_{i,1}$, next to $v_{i,2}$,
  and so on, until $v_{i,m}$. Each edge on this path has capacity
  equal to $\sum_{j=1}^m y_{i,j}$ (recall that this value is an
  integer).  For each $i, j \in [m]$, we have an edge from $v_{i,j}$
  to $t_j$, with capacity one and with cost $1-2y_{i,j}$ (all the
  other types of edges have cost $0$). Finally, for each $j \in [m]$,
  we have an edge from $t_j$ to $t$ with capacity
  $\sum_{i=1}^m y_{i,j}$ (this value is an integer). Next, we compute
  in polynomial time (using some classic algorithm) an integral flow
  that moves $\sum_{i,j\in[m]} y_{i,j}$ units of flow from $s$ to $t$
  at the lowest possible cost (which we denote $B_f$). If this flow
  existed, then we could compute matrix~$D$ by setting, for each
  $i, j \in [m]$, $d_{i,j}$ to be the amount of flow ($0$ or~$1$)
  going from $v_{i,j}$ to~$t_j$. Indeed, by definition of our flow
  network, for each~$i$ we would have that
  $\sum_{j=1}^m d_{i,j} = \sum_{j=1}^m y_{i,j}$ (because when a flow
  enters some node $v_{i,j}$, then it can either go to $t_j$ or to
  $v_{i,j+1}$). Due to the capacities on the edges from the pre-sink
  nodes to the sink, for each $j \in [m]$ we would also have that
  $\sum_{i=1}^m d_{i,j} = \sum_{i=1}^m y_{i,j}$.  These two properties
  are equivalent to ensuring that the degrees of the vertices in~$G'$
  are equal to the fractional degrees in~$G$, as is done by dependent
  rounding. Further, we have that:
  \[
    \sum_{1 \leq i,j \leq m} |nx_{i,j} - (\lfloor nx_{i,j} \rfloor + d_{i,j})| = B + B_f.
  \]
  To see why this is the case, note that if all the values $d_{i,j}$
  where $0$, then the left-hand sum would be $B$, and on the
  right-hand size $B_f$ would be~$0$.  Each $d_{i,j} = 1$ increases
  both sides of the sum by the same amount.
  
  It remains to see that the desired flow indeed always exists. However, this
  follows from the previous algorithm: Since dependent rounding always
  produces a desired matching, which induces matrix~$D$, $D$
  induces a flow of our desired value.    
\end{proof}

\section{Missing Material from \Cref{se:setcomp}}\label{ap:4}
\subsection{The Four Matrices}
We start this subsection by proving \Cref{pr:calc}. Afterwards, we provide some evidence that our four compass elections are, indeed, almost as far away from each other as possible. 
\dist*
\begin{proof}
\textbf{$\mathbf{\textbf{ID}_m}$ and $\mathbf{\textbf{UN}_m}$:} We start by computing the distance between $\ID_m$ and $\UN_m$. Note that $\UN_m$ always remains the same matrix independent of how its columns are ordered. Thus, we can compute the distance between these two matrices using the identity permutation between the columns of the two matrices: 
    $\POS(\ID_m,\UN_m) = \sum_{i=1}^m \EMD((\ID_m)_i, (\UN_m)_i)=\textstyle\sum_{i=1}^m (\textstyle\sum_{j=1}^{i-1} \frac{j}{m} + \textstyle\sum_{j=1}^{m-i} \frac{j}{m}) 
\\ = \frac{1}{m} \textstyle\sum_{i=1}^m ( \frac{1 + (i-1)}{2}(i-1) + \frac{1 + (m-i)}{2}(m-i) ) 
\\ = \frac{1}{2m} \textstyle\sum_{i=1}^m (2i^2 - 2i - 2mi + m^2  + m) 
\\ = \frac{1}{2m} (2\frac{m(m+1)(2m+1)}{6} - m(m+1)-m^2(m+1)  + m(m^2 + m))
\\ = \frac{1}{2m} (\frac{(m^2+m)(2m+1)}{3} - (m+1)(m+m^2)  + m(m^2 + m))
\\ = \frac{m+1}{2} (\frac{(2m+1)}{3} - (m+1)  + m) = \frac{(m+1)(m-1)}{3} = \frac{1}{3}(m^2-1).$

In the following, we use $(*)$ when we omit some calculations analogous to the calculations for $\POS(\ID_m,\UN_m)$.

\medskip
\noindent \textbf{$\mathbf{\textbf{UN}_m}$ and $\mathbf{\textbf{ST}_m}$:} Similarly, we can also directly compute the distance between $\UN_m$ and $\ST_m$ using the identity permutation between the columns of the two matrices. In this case, all column vectors of the two matrices have indeed the same $\EMD$ distance to each other:

 \noindent
 $\POS(\UN_m,\ST_m) =  m\cdot (\frac{1}{2}+2\cdot\textstyle\sum_{i=1}^{\frac{m}{2}-1} \frac{i}{m})= \frac{m}{2}+\frac{m}{2}(\frac{m}{2}-1) = \frac{m^2}{4}.$
 
  \medskip
\noindent \textbf{$\mathbf{\textbf{UN}_m}$ and $\mathbf{\textbf{AN}_m}$:} Next, we compute the distance between $\UN_m$ and $\AN_m$ using the identity permutation between the columns of the two matrices.  Recall that $\AN_m$ can be written as:
\[
  \AN_m = 0.5\begin{bmatrix}
    \ID_{\nicefrac{m}{2}} & \rID_{\nicefrac{m}{2}} \\
    \rID_{\nicefrac{m}{2}} & \ID_{\nicefrac{m}{2}}
  \end{bmatrix}.
\] Thus, it is possible to reuse our ideas from computing the distance between identity and uniformity:

  \noindent
 $\POS(\UN_m,\AN_m) = 4 \textstyle\sum_{i=1}^{\frac{m}{2}}(\textstyle\sum_{j=1}^{i-1} \frac{j}{m} + \textstyle\sum_{j=1}^{\frac{m}{2}-i} \frac{j}{m}) = (*) = \frac{2}{3}(\frac{m^2}{4}-1).$
 
 \medskip
\noindent \textbf{$\mathbf{\textbf{ID}_m}$ and $\mathbf{\textbf{ST}_m}$:}
There exist only two different types of column vectors in $\ST_m$, i.e.,  $\frac{m}{2}$ columns starting with $\frac{m}{2}$ entries of value $\frac{2}{m}$ followed by $\frac{m}{2}$ zero-entries and $\frac{m}{2}$ columns starting with $\frac{m}{2}$ zero entries followed by $\frac{m}{2}$ entries of value $\frac{2}{m}$. In $\ID_m$, $\frac{m}{2}$ columns have a one entry in the first $\frac{m}{2}$ rows and $\frac{m}{2}$ columns have a one entry in the last $\frac{m}{2}$ rows. Thus, again the identity permutation between the columns of the two matrices minimizes the $\EMD$ distance:

  \noindent
 $\POS(\ID_m,\ST_m) = 2\cdot \POS(\ID_{\frac{m}{2}},\UN_{\frac{m}{2}}) = \frac{2}{3}(\frac{m^2}{4}-1)$
 
  \medskip
\noindent \textbf{$\mathbf{\textbf{AN}_m}$ and $\mathbf{\textbf{ST}_m}$:} We now turn to computing the distance between $\AN_m=(\an_1,\dots , \an_m)$ and $\ST_m=(\stt_1,\dots , \stt_m)$. As all column vectors of $\AN_m$ are palindromes, each column vector of $\AN_m$ has the same $\EMD$ distance to all column vectors of $\ST_m$, i.e., for $i\in [m]$ it holds that $\EMD(\an_i,\stt_j)=\EMD(\an_i,\stt_{j'})$ for all $j,j'\in [m]$. Thus, the distance between $\AN_m$ and $\ST_m$ is the same for all permutation between the columns of the two matrices. Thus, we again use the identity permutation. 
We start by computing $\EMD(\an_i,\stt_i)$ for different $i\in [m]$ separately distinguishing two cases. Let $i\in [\frac{m}{4}]$. Recall that $\an_i$ has a $0.5$ at position $i$ and position $m-i+1$ and that $\stt_i$ has a $\frac{2}{m}$ at entries $j\in [\frac{m}{2}]$. We now analyze how to transform $\an_i$ to $\stt_i$. For all $j\in [i-1]$, it is clear that it is optimal that the value $\frac{2}{m}$ moved to position $j$ comes from position $i$. The overall cost of this is $\textstyle\sum_{j=1}^{i-1} \frac{2j}{m}$. Moreover, the remaining surplus value at position $i$ (that is, $\frac{1}{2}-\frac{2i}{m}$) needs to be moved toward the end. Thus, for $j\in [i+1,\frac{m}{4}]$, we move value $\frac{2}{m}$ from position $i$ to position $j$. The overall cost of this is  $\textstyle\sum_{j=1}^{\frac{m}{4}-i} \frac{2j}{m}$. Lastly, we need to move value $\frac{2}{m}$ to positions $j\in [\frac{m}{4}+1,\frac{m}{2}]$. This needs to come from position $m-i+1$. Thus, for each $j\in [\frac{m}{4}+1,\frac{m}{2}]$, we move value $\frac{2}{m}$ from position $m-i+1$ to position $j$. The overall cost of this is $\frac{1}{2}\cdot (\frac{m}{2}-i)+\textstyle\sum_{j=1}^{\frac{m}{4}} \frac{2j}{m}=\frac{1}{2}(\frac{m}{2}-i)+\frac{m}{16}+\frac{1}{4}$ 

Now, let $i\in [\frac{m}{4}+1,\frac{m}{2}]$. For $j\in [\frac{m}{4}]$, we need to move value $\frac{2}{m}$ from position $i$ to position $j$. The overall cost of this is $\frac{1}{2}\cdot (i-\frac{m}{4}-1)+\textstyle\sum_{j=1}^{\frac{m}{4}} \frac{2j}{m}=\frac{1}{2}\cdot (i-\frac{m}{4}-1)+\frac{m}{16}+\frac{1}{4}$. For $j\in [\frac{m}{4}+1,\frac{m}{2}]$, we need to move value $\frac{2}{m}$ from position $m-i+1$ to position $j$. The overall cost of this is $\frac{1}{2}\cdot (\frac{m}{2}-i)+\textstyle\sum_{j=1}^{\frac{m}{4}} \frac{2j}{m}=\frac{1}{2}\cdot (\frac{m}{2}-i)+\frac{m}{16}+\frac{1}{4}$. 

Observing that the case $i\in [\frac{3m}{4}+1,m]$ is symmetric to $i\in [\frac{m}{4}]$ and the case $i\in [\frac{m}{2}+1,\frac{3m}{4}]$ is symmetric to $i\in [\frac{m}{4}+1,\frac{m}{2}]$ the $\EMD$ distance between $\AN_m$ and $\ST_m$ can be computed as follows:

\noindent$\POS(\AN_m,\ST_m) = 2\cdot ( A + \frac{1}{2}\cdot (\sum_{i=1}^{\frac{m}{4}} \frac{m}{2}-i)  + \frac{m}{4}\cdot (\frac{m}{16}+\frac{1}{4})+ \frac{1}{2} \cdot (\sum_{i=\frac{m}{4}+1}^{\frac{m}{2}} \cdot (i-\frac{m}{4}-1))+\frac{m}{4}\cdot (\frac{m}{16}+\frac{1}{4})+ \frac{1}{2}\cdot (\sum_{i=\frac{m}{4}+1}^{\frac{m}{2}} \frac{m}{2}-i)+\frac{m}{4}\cdot (\frac{m}{16}+\frac{1}{4}))\\
=\frac{m^2}{48} - \frac{1}{3} + \frac{3m^2-4m}{32}  + \frac{m}{2}\cdot (\frac{m}{16}+\frac{1}{4})+  \frac{m^2-4m}{32}+\frac{m}{2}\cdot (\frac{m}{16}+\frac{1}{4})+ \frac{m^2-4m}{32}+\frac{m}{2}\cdot (\frac{m}{16}+\frac{1}{4}))\\
=\frac{m^2}{48}-\frac{1}{3}+\frac{3m^2-4m}{32}+\frac{3m}{2}(\frac{m}{16}+\frac{1}{4})+\frac{m^2-4m}{16}=(\frac{1}{48}+\frac{3}{32}+\frac{3}{32}+\frac{1}{16})m^2+(-\frac{4}{32}+\frac{3}{8}-\frac{4}{16})m\frac{1}{3}=\frac{13}{48}m^2-\frac{1}{3}$

\smallskip
\noindent with 

\noindent$A =  \textstyle\sum_{i=1}^{\frac{m}{4}}(\textstyle\sum_{j=1}^{i-1} \frac{2j}{m} + \textstyle\sum_{j=1}^{\frac{m}{4}-i} \frac{2j}{m}) = (*) = \frac{1}{6}(\frac{m^2}{16}-1) = \frac{1}{2}(\frac{m^2}{48} - \frac{1}{3})$

  \medskip
  
\noindent \textbf{$\mathbf{\textbf{ID}_m}$ and $\mathbf{\textbf{AN}_m}$:} Lastly, we consider $\ID_m=(\id_1,\dots , \id_m)$ and $\AN_m=(\an_1,\dots , \an_m)$. Note that, for $i\in [m]$, $\id_i$ contains a $1$ at position $i$ and $\an_i$ contains a $0.5$ at position $i$ and position $m-i$. Note further that for $i\in [\frac{m}{2}]$ it holds that $\an_i=\an_{m-i+1}$.
Fix some $i\in [\frac{m}{2}]$. For all $j\in [i,m-i+1]$ it holds that $\EMD(\an_i,\id_j)=\frac{m-2i+1}{2}$ and for all $j\in [1,i-1]\cup [m-i+2,m]$ it holds that $\EMD(\an_i,\id_j)>\frac{m-2i+1}{2}$. 
That is, for every $i\in [m]$, $\an_i$ has the same distance to all column vectors of $\ID_m$ where the one entry lies in between the two $0.5$ entries of $\an_i$ but a larger distance to all column vectors of $\ID_m$ where the one entry is above the top $0.5$ entry of $\an_i$ or below the bottom $0.5$ entry of $\an_i$. Thus, it is optimal to choose a mapping of the column vectors such that for all $i\in [m]$ it holds that $\an_i$ is mapped to a vector $\id_j$ where the one entry of $\id_j$ lies between the two $0.5$ in $\an_i$. This is, among others, achieved by the identity permutation, which we use to compute:

 \noindent
 $\POS(\ID_m,\AN_m) = 2 \textstyle\sum_{i=1}^{\frac{m}{2}} (\frac{1}{2} (m-2i+1)) 
 \\ = \frac{m}{2}m-\frac{m}{2}(\frac{m}{2}+1)+\frac{m}{2} = \frac{m^2}{4}$
\end{proof}

\bigskip
Focusing on $6\times 6$ matrices, we now argue that our four compass matrices are quite far away from each other. We focus on $m=6$, as this is the largest dimension for which the ILPs we will use in this section were able to compute optimal solutions within several hours. While it could be the case that for larger $m$ the results are significantly different, we do not expect that this is the case (see e.g. \Cref{sec:four}). Moreover, the case $m=6$ is already interesting on its own, as elections with only six candidates also regularly appear in the real world.

We present two different justifications that our compass matrices cover very different areas on the map. First, we explain how we created the compass by adding the four matrices one after another, trying to maximize the distances between them. Second, we argue that each of the matrices is almost as far away from the other three as possible.

\subsubsection{Iteratively Building the Compass}
We built the compass in an iterative fashion, i.e., we added the matrices one after each other. 
We first wanted to find the two matrices that are furthest away from each other. As discussed in \Cref{sec:four}, we believe that these two matrices are the identity and uniformity matrices. For $m=6$, using an ILP, we were able to verify that, indeed, $\ID_6$ and $\UN_6$ have the highest distance among all pairs of $6 \times 6$ frequency matrices. 

Next, we wanted to find a frequency matrix that is as far away from $\ID_6$ and $\UN_6$ as possible.  However, in this context, it is not entirely clear what ``as far as possible'' means, as one might, e.g., maximize the sum or the minimum of the distances to $\ID_6$ and $\UN_6$.
Using again an ILP, we found two (quite unstructured) matrices $M$ (with maximum minimum distance) and $S$ (with maximum summed distance): 
\begin{align*}M=
  \begin{bmatrix}
     0.55& 0.45& 0& 0& 0& 0 \\
     0.05&0.25 & 0.7&0 & 0& 0 \\
     0& 0.16& 0&0.6 & 0.24& 0 \\
     0& 0.14& 0.3& 0& 0.16& 0.4 \\
     0& 0& 0&0 &0.4 & 0.6 \\
     0.4& 0& 0& 0.4& 0.2& 0 \\
  \end{bmatrix} 
\end{align*}
\begin{align*}S=
  \begin{bmatrix}
    0.33 & 0.33 & 0.33 & 0 & 0 & 0 \\
    0.08 & 0.08 & 0.08 & 0.17 & 0.17 & 0.42 \\
    0 & 0 & 0 & 0.5 & 0.5 & 0 \\
    0.08 & 0.08 & 0.8 & 0.33 & 0.33 & 0.08 \\
    0.17 & 0.17 & 0.5 & 0 & 0 & 0.17 \\
    0.33 & 0.33 & 0 & 0 & 0 & 0.33 \\
  \end{bmatrix}.
\end{align*}

The distances of $M$ and $S$ to $\ID_6$ and $\UN_6$ can be found in 
\Cref{tab:dist}. As both $M$ and $S$ are quite unstructured, it is unclear how 
to generalize them to higher dimensions, and it is not intuitively clear what 
kind of elections they resemble, we started thinking about canonical matrices 
distant from identity and uniformity and came up with $\AN$ and $\ST$. As 
displayed in \Cref{tab:dist}, both $\AN_6$ and $\ST_6$ are quite far away from 
both $\ID_6$ and $\UN_6$. Overall, the summed distance of these two matrices to 
$\ID_6$ and $\UN_6$ is close to the maximum achievable distance.
We decided to first add antagonism to the compass. 

\begin{table}
\centering
  \begin{tabular}{l|ccc}
    \toprule
    & $\POS(\ID_6,\cdot)$ & $\POS(\UN_6,\cdot)$ &$\POS(\UN_6,\cdot)$\\
    &  & &+$\POS(\UN_6,\cdot)$ \\
    \midrule
    $M$  & 7.18 & 7.18  & 14.36 \\
   $S$  & 10 & 4.66 & 14.66 \\
    $\AN_6$& 9 & 5.33 & 14.33\\
    $\ST_6$& 9 & 5.33 & 14.33\\
    \bottomrule
  \end{tabular}
  \caption{\label{tab:dist} Distance of four $6\times 6$ matrices to $\ID_6$ and $\UN_6$. $M$ is the matrix with maximum possible minimum distance and $S$ the matrix with maximum possible summed distance to $\ID_6$ and $\UN_6$.}
\end{table}
In the next step, we  again computed two matrices $M'$ with maximum minimum distance from $\ID_6$, $\UN_6$, and $\AN_6$, and $S'$ with maximum summed distance from $\ID_6$, $\UN_6$, and $\AN_6$. It turned out that $S'=\ST_6$ and
\begin{align*}M'=
  \begin{bmatrix}
  0.33 & 0.33 & 0.33 & 0 & 0 & 0 \\
    0.33 & 0.33 & 0.33 & 0 & 0 & 0 \\
   0 & 0 & 0.08 & 0.25 & 0.33 & 0.33 \\
    0.33 & 0.33 & 0.25 & 0.08 & 0& 0 \\
    0 & 0 & 0 & 0.33 & 0.33 & 0.33 \\
    0 & 0 & 0 & 0.33 & 0.33 & 0.33 \\
  \end{bmatrix}.
\end{align*}

As depicted in \Cref{tab:dist2}, $\ST_6$ is quite far away from each of the three already fixed matrices from our compass. As, at the same time, $\ST_6$ maximizes the summed distance among all matrices, we selected stratification as our fourth compass matrix.
 
\subsubsection{Distance of One Compass Matrix from the Other Three}
We now describe a second experiment to verify that our compass matrices are indeed nearly as far away from each other as possible. For each of the four matrices $X\in \{\ID_6,\UN_6,\AN_6,\ST_6\}$, using an ILP, we compute the best matrix to replace $X$ in order to maximize the diversity of the resulting set of four matrices. That is, 
we compute the $6\times 6$ matrix with maximum summed distance to the matrices in $\{\ID_6,\UN_6,\AN_6,\ST_6\}\setminus \{X\}$ and the $6\times 6$ matrix with maximum minimum distance to a matrix in $\{\ID_6,\UN_6,\AN_6,\ST_6\}\setminus \{X\}$. The results of this experiment along with the minimum/summed distance of $X$ to the matrices in $\{\ID_6,\UN_6,\AN_6,\ST_6\}\setminus \{X\}$ are shown in \Cref{tab:div}. The results show that each compass matrix has the highest possible or close to the highest possible summed distance to the other three matrices. This implies that none of our matrices can be replaced by another matrix such that the resulting set covers a significantly larger area of the map. Concerning the minimum distance of each compass matrix to the other three, the compass matrices are no longer very close to being optimal. Nevertheless, each pair of matrices is quite far away from each other. Moreover, maximizing the minimum and summed distance are (at least partly) conflicting optimization goals and, for our purpose of putting a compass on the map, the summed distance is of greater importance.

\begin{table}
\centering
  \begin{tabular}{l|cccc}
    \toprule
    & $\POS(\ID_6,\cdot)$ & $\POS(\UN_6,\cdot)$ & $\POS(\AN_6,\cdot)$& $\POS(\ID_6,\cdot)$\\
    &  & & &+$\POS(\UN_6,\cdot)$ \\
    &  & & &+$\POS(\AN_6,\cdot)$ \\
    \midrule
    $M'$  & 7.17 & 7.16  & 7.82 & 22.15 \\
    $\ST_6$& 9 & 5.33 & 9.66 & 23.99\\
    \bottomrule
  \end{tabular}
  \caption{\label{tab:dist2} Distance of two $6\times 6$ matrices to $\ID_6$, $\UN_6$, and $\AN_6$. $M'$ is the matrix with maximum possible minimum distance to the three matrices.}
\end{table}
\subsection{Paths Between Election Matrices}
\paths*
\begin{proof}
  Let $Z = (z_1, \ldots, z_m) = \alpha X + (1-\alpha) Y$ be our convex
  combination of $X$ and $Y$.  We note two properties of the earth
  mover's distance. Let $a$, $b$, and $c$ by three vectors that
  consist of nonnegative numbers, where the entries in $b$ and $c$ sum up to the same
  value. Then, it holds that $\EMD(a+b,a+c) = \EMD(b,c)$. Further, for a
  nonnegative number~$\lambda$, we have that
  $\EMD(\lambda b, \lambda c) = \lambda\EMD(b,c)$.  Using these
  observations and the definition of the earth mover's distance, we
  note that:
  \begin{align*}
    \textstyle
    \POS(X,Z)  & \textstyle\leq \sum_{i=1}^m \EMD(x_i,z_i) \\
    &\!\!\!\!\!\!\!\!\!\!\!\!\!\!\!\! =   \textstyle\sum_{i=1}^m \EMD(x_i,\alpha x_i + (1-\alpha)y_i) \\
    &\!\!\!\!\!\!\!\!\!\!\!\!\!\!\!\! =   \textstyle\sum_{i=1}^m \EMD((1-\alpha)x_i, (1-\alpha)y_i) \\
    &\!\!\!\!\!\!\!\!\!\!\!\!\!\!\!\! =    \textstyle(1-\alpha) \sum_{i=1}^m \EMD(x_i,y_i) = (1-\alpha)\POS(X,Y).
  \end{align*}
  The last equality follows by our assumption regarding $X$ and
  $Y$. By an analogous reasoning we also have that
  $\POS(Z,Y) \leq \alpha \POS(X,Y)$. By putting these two inequalities
  together, we have that:
  \[
    \POS(X,Z) + \POS(Z,Y) \leq \POS(X,Y).
  \]
  By the triangle inequality, we have that $\POS(X,Y) \leq \POS(X,Z) + \POS(Z,Y)$
  and, so, we have that
  $\POS(X,Z) + \POS(Z,Y) = \POS(X,Y)$.
\end{proof}

\section{Missing Material from \Cref{se:appcomp}}
\subsection{Further Details on Normalized Mallows}
\label{ap:mallows}
We start this subsection by providing the missing computational details for our new parameterization that we propose in the main body. Afterwards, we look at where Mallows elections for a specific value of $\relphi$ end up on the map of elections for different numbers of candidates. 
\begin{table*}
\centering
  \begin{tabular}{l|cc|cc}
    \toprule
    & Minimum& Summed & Maxmimum minimum & Maxmimum summed\\
    &distance& distance & distance & distance\\
    \midrule
    $\ID$  & 5.33 & 26  & 7.75 & 26 \\
    $\UN$ & 5.33 & 26 & 7.63 & 26.44\\
    $\AN$  & 5.33 & 23.99   & 7.05 & 24.17 \\
    $\ST$  & 5.33 & 23.99  & 7.16 & 23.99 \\
    \bottomrule
  \end{tabular}
  \caption{\label{tab:div} For each of the four compass matrices (with m=6), the first column shows its minimum distance to the other three matrices and the second column its summed distance to the other three matrices. The third columns shows the maximum minimum distance of any matrix to the other three matrices and the fourth column the maximum summed distance of any matrix to the other three matrices.}
\end{table*}

\subsubsection{Computational Details}
We start with some definitions. Let $m\in \mathbb{N}$ and
$\phi\in [0,1)$. Moreover, recall that we denote by $\expswaps(m,\phi)$ the
expected swap distance between a reference vote~$v^*$ and a vote sampled
from Mallows model with dispersion parameter~$\phi$ and reference vote $v^*$ for $m$ candidates.
We denote by $\relswaps(m,\phi)$ the relative expected swap distance,
that is, the expected swap distance $\expswaps(m,\phi)$ normalized by
the maximum possible number $\frac{m(m-1)}{2}$ of swaps in a vote over
$m$ candidates.

In the main part, we proposed a new normalization of Mallows model
which, for a given $\relphi \in [0,1)$, crucially relies on finding a
$\phi \in [0,1)$ such that:
\begin{equation} \label{eq:mal1}
    \relswaps(m,\phi) = \relphi.
  \end{equation} We now describe how $\phi$ can be computed from $\relphi$. As a first step, we compute $\expswaps(m,\phi)$ for some given $m\in \mathbb{N}$ and $\phi\in [0,1)$ . Recall that, given $\phi\in [0,1)$ and a reference vote $v^*$ over $m$ candidates, the probability of sampling a vote $v$ over $m$ candidates under the Mallows model is:
\begin{equation}\label{eq:mal}
   \mathbb{P}_{\phi,v^*}(v)=\frac{1}{Z}\phi^{\kappa(v,v^*)} 
\end{equation}
with normalizing constant
$Z=1\cdot (1+\phi)\cdot (1+\phi+\phi^2)\cdot \dots
\cdot(1+\dots+\phi^{m-1})$.  Next, we need to compute the number of
different votes at some swap distance from a given vote over $m$ candidates. For this, we
create a table $T$ and let $T[m,i]$ denote the number of different
votes at some given swap distance $i\in [\frac{m(m-1)}{2}]$ from some
fixed vote over $m$ candidates. Notably, $T[m,i]$ corresponds to the
number of permutations over $m$ elements containing $i$ inversions, a
well-studied combinatorial problem (\cite{oeis}). Unfortunately, no
closed form expression to compute this is known. Instead, one needs to
resort to dynamic programming (\cite{oeis}): First, we initialize the
table with $T[m,0] = 1$. Subsequently, for increasing $m$, we update
the table starting from $i= 0$ and going to $i=\frac{m(m-1)}{2}$ as:
$$T[m,i]=T[m,i-1]+T[m-1,r]-T[m-1,r-m].$$

After precomputing $T$, we are able to compute $\expswaps(\phi,m)$ using \Cref{eq:mal} as: 
\begin{equation}
    \expswaps(\phi,m)=\frac{1}{Z}\sum_{i=0}^{\frac{m(m-1)}{2}} T[m,i]\cdot \phi^i.
\end{equation}
The idea behind this equation is that, for each $i\in [\frac{m(m-1)}{2}]$, there exist $T[m,i]$ different votes at swap distance $i$ from the reference vote, each being sampled with probability $\frac{1}{Z}\phi^i$. Using this, we are now able to reformulate \Cref{eq:mal1}: Given some $m\in \mathbb{N}$ and $\relphi \in [0,1)$, find $\phi$ such that:
$$-1+\frac{2}{\relphi \cdot Z \cdot m(m-1)}\sum_{i=0}^{\frac{m(m-1)}{2}} T[m,i]\cdot \phi^i=0.$$ Note that the left hand side of this equation is simply a polynomial in $\phi$ of maximum degree $\frac{m(m-1)}{2}$. Thus, solving \Cref{eq:mal1}  reduces to finding the root of a higher degree polynomial. While no general formula for finding such a root exists, there exists a whole branch of literature focusing on numerical methods for solving this problem, e.g., Newton's method (\cite{DBLP:books/lib/PressTVF07}, Chapter 9).\footnote{In our python implementation, we used the root() method from the numpy library, which relies on computing the eigenvalues of the companion matrix (\cite{DBLP:books/cu/HJ2012}).}

\begin{figure}
  \centering
  \begin{subfigure}[t]{6cm}
    \centering
    \includegraphics[width=5.7cm]{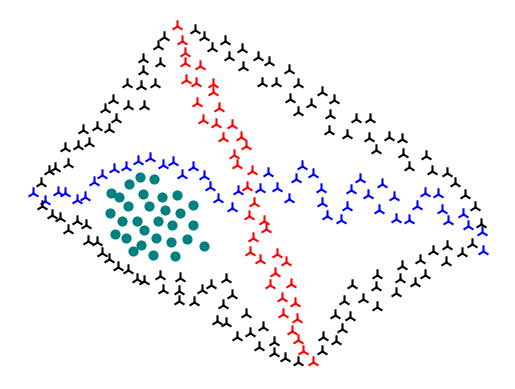}
    \caption{\label{fi:mal1} Mallows elections for $m=5$.}
  \end{subfigure}
  \quad
  \begin{subfigure}[t]{6cm}
    \centering
    \includegraphics[width=5.7cm]{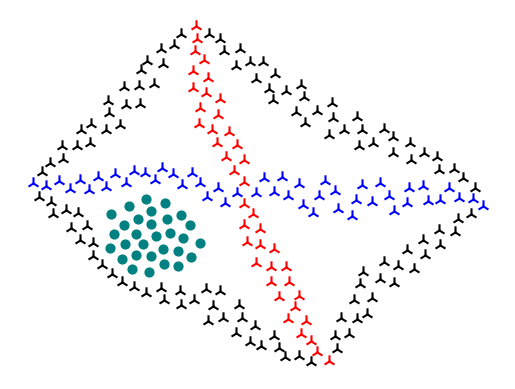}
    \caption{\label{fi:mal2}Mallows elections for $m=10$.}
  \end{subfigure}

  \begin{subfigure}[t]{6cm}
    \centering
    \includegraphics[width=5.7cm]{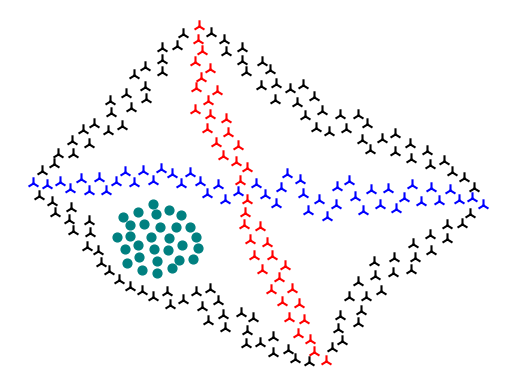}
    \caption{\label{fi:mal3} Mallows elections for $m=20$.}
  \end{subfigure}
  \quad
  \begin{subfigure}[t]{6cm}
    \centering
    \includegraphics[width=5.7cm]{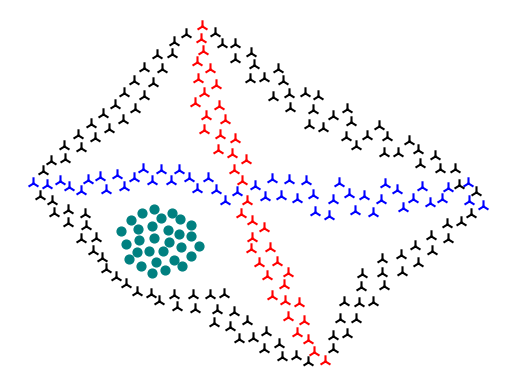}
    \caption{\label{fi:mal4} Mallows elections for $m=50$.}
  \end{subfigure}
  \quad

  \caption{\label{fi:malmaps}Visualization of our four compass matrices and the connecting paths for different numbers of candidates. Additionally, several elections randomly generated from Mallows model with parameter $\relphi=0.375$ for the respective number of candidates are included.}
\end{figure}
\subsubsection{Mallows Elections for Different Numbers of Candidates}\label{ssub:malm}

As we only looked at the special case $m=10$, it is, in principle questionable whether our recommendations which dispersion parameter to choose presented in \Cref{sub:recom} are also valid for $m>10$ candidates. As mentioned in the main body, we suspect that real-life elections of similar types end up at similar places on the map of elections with our compass on it for different numbers of candidates. Thus, it remains to check whether elections generated from Mallows model with some fixed $\relphi$ for different numbers of candidates also end up in similar positions on the respective maps. We confirm this hypothesis using \Cref{fi:malmaps}, where we depict the map of elections for $m\in \{5,10,20,50\}$ including the four compass matrices, the paths between them, and several elections generated using Mallows model with $\relphi=0.375$ (we look at $\relphi=0.375$ because this is the best parameter to model the quite homogeneous group of political elections). All maps look quite similar and elections generated using the Mallows model indeed always end up in very similar places. 

\subsection{Generation, Selection, and Preprocessing of Real-Life Datasets} \label{ap:datasets}

\subsubsection{Generation and Selection of Datasets}
As part of our work, we generate three new datasets from sports competitions, that are,  Tour de France elections, Giro d'Italia elections, and speed skating elections. As we are interested in elections with complete votes, as a preprocessing step, we always delete votes and candidates until each candidate is contained in at least 70\% of all votes and each vote contains 70\% of all candidates.

Both the Tour de France and Giro d'Italia are annual cycling competitions consisting of multiple stages. For each of the past 100 editions, we create a separate election with the riders as the candidates and the stages as the voters (where each voter ranks the candidates according to the finish times of the riders in the corresponding stage). To generate the elections, we use publicly available data from \url{https://www.procyclingstats.com}. The third new dataset consists of different elections modelling speed skating competitions. For this, we use data from \url{https://results.sporthive.com/} and select 51 speed skating races. Notably, each race consists of multiple laps. For each race, we create a separate election with the speed skaters as the candidates and the laps as the voters (where each voter ranks the candidates according to the lap time of the speed skaters in the corresponding lap).

In addition, our main source of real-life elections is the PrefLib database (\cite{mat-wal:c:preflib}). We categorize all election datasets from PrefLib in \Cref{tab:preflib}. As discussed in the main body, we want to compare elections with ten candidates. Moreover, as our model only allows to consider complete votes without ties, we are interested in datasets where votes are as complete as possible and  contain only a few ties. Based on these criteria, our decision which PrefLib datasets to include is displayed in \Cref{tab:preflib}.

\subsubsection{Preprocessing of Datasets}
For all elections from our selected datasets containing incomplete
votes (i.e., votes where some of the top candidates are ranked and the
remaining candidates are not), we need to fill-in the votes. For the
decision how to complete each vote, we use the other votes as
references assuming that voters that rank the same candidates on top
also continue to rank candidates similarly towards the bottom. For
each incomplete vote $v$, we proceed as follows. Let us assume that
the length of the vote~$v$ is~$n$. Let $V_P$ be the set of all
original votes of which $v$ is a prefix. We uniformly at random select
one vote $v_p$ from $V_P$ and then we add candidate $c=\pos(v_p, n+1)$
at the end of vote~$v$. We repeat the procedure until vote~$v$ is
complete. If the set~$V_P$ is empty, then we choose $c$ uniformly at
random (from those candidates that are not part of $v$ yet).
Moreover, if a vote contains ties (i.e., pairs or larger sets of
candidates that are reported as equally good, except for the case that
would fit the description of an incomplete vote above), we break them
randomly.

After applying these preprocessing steps, we arrive at a collection of datasets containing elections with ten or more candidates and complete votes without ties. As we focus on ten candidates, we need to select a subset of ten candidates for each election: We select the ten candidates with the highest Borda score.

\begin{table*}[t]
\centering
\resizebox{\textwidth}{!}{\begin{tabular}{c  c  c  c  c  c  c   c}
			\toprule
			PrefLib ID & Name & Size & $m$ & Type & Selected &  Reason to reject \\	
			\midrule
			1 & Irish & 3 & 9,12,14 & soi & yes &  -\\
			2 & Debian & 8 & 4-9 & toc & no & Too few candidates\\
			3 & Mariner & 1 & 32 & toc & no & Too many ties\\
			4 & Netflix & 200 & 3,4 & soc & no &Too few candidates \\
			5 & Burlington & 2 & 6 & toi & no & Too few candidates  \\
			6 & Skate & 48 & 14-30 & toc & yes &  -\\
			7 & ERS	& 87 & 3-29 & soi & yes &  -\\
			8 & Glasgow	& 21 & 8-13 & soi & yes&- \\
            9 & AGH & 2 & 7,9 & soc & no &  Too few candidates \\
            10.1 & Formula & 48 & 22-62 & soi & no &  Incomplete and few votes\\
            10.2 & Skiing & 2 & $\sim$50 & toc & no & Few votes and many ties \\
            11.1 & Webimpact & 3 &	103, 240, 242  & soc & no & Too many candidates and too few votes ($\sim$5)\\	
            11.2 & Websearch & 74 &	100-200, $\sim$2000 & soi & no & Too few votes ($\sim$4)	\\				
            12 & T-shrit & 1 & 11 & soc & yes &  -\\
            13 & Anes &	19 & 3-12 & toc &	no & Too many ties\\
            14 & Sushi & 1 & 10 & soc & yes & -\\
            15 & Clean Web & 79 & 10-50, $\sim$200 & soc & no & Too few votes ($\sim$4)	\\		
            16 & Aspen & 2 & 5,11 & toc & yes & 		-\\		
            17 & Berkeley & 1 & 4 & toc & no & Too few candidates 	\\			
            18 & Minneapolis & 4 & 7,9,379,477 & soi & no & Incomplete votes	\\			
            19 & Oakland & 7 & 4-11 & toc & no & Incorrect data (votes like: 1,1,1)\\
            20 & Pierce	& 4 & 4,5,7 & toc & no & Too few candidates			\\
            21 & San Francisco & 14 & 4-25 & toc & no & Incorrect data (votes like: 1,1,1)\\
            22 & San Leonardo & 3 & 4,5,7 & toc & no & Too few candidates \\ 
            23 & Takoma	& 1 & 4 & toc & no & Too few candidates \\
            24 & MT Dots & 4 & 4 & soc & no & Too few candidates					\\
            25 & MT Puzzles	& 4 & 4 & soc & no & Too few candidates				\\
            26 & Fench Presidential &	6 & 16 & toc & no & Approval ballots			\\	
            27 & Proto French & 1 & 15 & toc & no & Approval ballots					\\
            28 & APA & 12 & 5 & soi & no & Too few candidates				\\
            29 & Netflix NCW & 12 & 3,4 & soc & no & Too few candidates						\\
            30 & UK labor party	& 1 & 5 & soi & no & Too few candidates				\\
            31 & Vermont & 15 & 3-6 & toc & no & Approval ballots \\
            32 & Cujae	& 7 & 6,32 & soc/soi/toc & no & Many reasons				\\
            33 & San Sebastian Poster & 2 & 17 & toc & no & Approval ballots \\
            34 & Cities survey & 2 & 36, 48 & soi & yes &  -\\				
			\bottomrule
	\end{tabular}}
	\caption{\label{tab:preflib} Overview of all election datasets that are part of the PrefLib database.  ``Size'' stands for the number of elections in the dataset,  ``$m$'' for the number of candidates and  ``Type'' for the type of the votes in the dataset (soc means that all votes are strict complete orders; soi means that all votes are strict incomplete orders; toc means that all votes are weak incomplete orders).}
      \end{table*}

In \Cref{tab:preflib_selected} we present a detailed description of the resulting selected datasets. In some datasets only parts of the data meets our criteria. For example, in the dataset containing Irish elections we have three different election, but one of the them (an election from West Dublin) contains only 9 candidates. We delete all such elections. After doing so, we finally arrive at eleven real-life datasets containing elections meeting our criteria.

However, as we cannot include all elections from each dataset on the map of elections, we further reduce the number of elections by considering only selected elections. In \Cref{tab:preflib_selected}, we include in the column ``\# Selected Elections'' the number of elections we selected from each dataset in the end. We based our decision on the selection of elections on the number of voters and candidates. That is, for ERS, we only take election with at least $500$ voters, for Speed Skating with at least $80$, for TDF with at least $20$, and for Figure skating with at least $9$. Beside that for TDF, we only select elections with no more than $75$ candidates. We refer to the resulting datasets as \emph{intermediate} datasets. 

\subsubsection{Sampling Elections from Intermediate Datasets}
We treat each of our intermediate real-life datasets as a separate election model from which we sample $15$ elections to create the final datasets that we use in the paper. For each intermediate dataset, we sample elections as follows. First, we randomly select one of the elections. Second, we sample $100$ votes from the election uniformly at random (this implies that for elections with less than $100$ votes, we select some votes multiple times, and for elections with more than $100$ votes, we do not select some votes at all). We do so to make full use of elections with far more than $100$ votes. For instance, our Sushi intermediate dataset contains only one election consisting of $5000$ votes. Sampling an election from the Sushi  intermediate dataset thus corresponds to drawing $100$ votes uniformly at random from the set of $5000$ votes. On the other hand, for intermediate datasets containing a higher number of elections, e.g., the Tour de France intermediate dataset, most of the sampled elections come from a different original election.

After executing this procedure, we arrive at $11$ datasets each containing $15$ elections consisting of $100$ complete and strict votes over $10$ candidates, which we use for our experiments.

\begin{table*}[t]
\centering
\resizebox{\textwidth}{!}{\begin{tabular}{c  c  c  c  c  c  c}
			\toprule
			Source & Category & Name & \# Selected Elections & Avg $m$ & Avg $n$ & Description\\	
			\midrule
			Preflib & Political & Irish & 2 & 13 & $\sim$ 54011 & Elections from North Dublin and Meath\\
			Preflib & Political & Glasgow	& 13 & $\sim$ 11 & $\sim$ 8758 & City council elections \\
            Preflib & Political & Aspen & 1 & 11 & 2459	& City council elections\\
			Preflib & Political & ERS	& 13 & $\sim$ 12 & $\sim$ 988 & Various elections held by non-profit organizations,\\ 
			        & & & & & & trade unions, and professional organizations  \\
            \midrule
			Preflib & Sport & Figure Skating & 40 & $\sim$ 23 & 9 & Figure skating  \\
			This paper & Sport & Speed Skating & 13 & $\sim$ 14 & 196 & Speed skating  \\
            This paper & Sport & TDF & 12 & $\sim$ 55 & $\sim$ 22 & Tour de France\\	
            This paper & Sport & GDI & 23 & $\sim$ 152 & 20  & Giro d’Italia	\\	
            \midrule
            Preflib & Survey & T-Shirt & 1 & 11 & 30 & Preferences over T-Shirt logo \\	
            Preflib & Survey & Sushi & 1 & 10 & 5000 & Preferences over Sushi\\	
            Preflib & Survey & Cities & 2 & 42 & 392 & Preferences over cities	\\			
			\bottomrule
	\end{tabular}}
	\caption{\label{tab:preflib_selected} Each row contains a description of one of the real-life datasets we consider. In the column ``\# Selected Elections'', we denote the number of elections we finally select from the respective dataset.}
\end{table*}
\end{document}